\RequirePackage{amsmath}
\documentclass[letterpaper,11pt,english]{article}
\usepackage[letterpaper, portrait, margin=2.5cm]{geometry}

\usepackage{amssymb, amsfonts}
\usepackage{graphicx}
\usepackage{amsthm}
\usepackage{algorithm}
\usepackage[noend]{algpseudocode}
\usepackage[nolist]{acronym}
\usepackage[capitalize,nameinlink,noabbrev]{cleveref}
\usepackage{todonotes}
\usepackage{mathtools}
\usepackage{thm-restate}
\newacro{ses}[SES]{smallest enclosing sphere}
\usepackage{csquotes}
\usepackage[inline,shortlabels]{enumitem}

\usepackage{todonotes}
\usepackage[compatibility=false]{caption}
\usepackage{subcaption}

\newtheorem{theorem}{Theorem}
\newtheorem{lemma}[theorem]{Lemma}

\newtheorem{corollary}[theorem]{Corollary}
\newtheorem{definition}{Definition}

\def\Fsync/{$\mathcal{F}$\textsc{sync}}
\def\Gathering/{\textsc{Gathering}}
\def\chainForm/{\textsc{Chain-Formation}}
\def\patternForm/{\textsc{Pattern Formation}}
\def\Point/{\textsc{Point}}
\def\UniformCircle/{\textsc{Uniform Circle}}
\def\gtc/{\textsc{Go-To-The-Center}}
\def\gtcShort/{\textsc{GTC}}
\def\mobs/{\textsc{Move-on-Bisector}}
\def\gtcThreeD/{\textsc{3d-Go-To-The-Center}}
\def\gtcThreeDShort/{\textsc{3d-GTC}}
\def\gtcThreeDCont/{\textsc{Continuous-3d-Go-To-The-Center}}
\def\gtcThreeDContShort/{\textsc{Cont-3d-GTC}}
\def\tangentialNormal/{tangential-normal}
\def\moveOnAngleMinimizer/{\textsc{Move-on-Angle-Minimizer}}
\def\look/{\texttt{Look}}
\def\compute/{\texttt{Compute}}
\def\move/{\texttt{Move}}

\def\LCMlong/{\textsc{Look-Compute-Move}}
\def\LCM/{LCM}

\def\chainForm/{\textsc{Chain-Formation}}
\def\Gathering/{\textsc{Gathering}}
\def\maxForm/{\textsc{Max-Chain-Formation}}

\def\gtm/{\textsc{GtM}}
\def\gtmlong/{\textsc{Go-To-The-Middle}}
\def\mob/{\textsc{MoB}}
\def\moblong/{\textsc{Move-On-Bisector}}
\def\maxmob/{\textsc{Max-MoB}}
\def\naivemaxmob/{\textsc{Naive-Max-MoB}}
\def\maxmoblong/{\textsc{Max-Move-On-Bisector}}

\def\Oblot/{\ensuremath{\mathcal{OBLOT}}}

\newcommand{\parameterizedMaxMoveOnBisectorLong}[2]{\textsc{Max-Move-On-Bisector}}
\newcommand{\parameterizedMaxMoveOnBisector}[2]{\textsc{Max-MoB}}
\newcommand{\AName}{\textsc{Max-GtM}}
\def\maxgtm/{\AName}
\def\maxgtmlong/{\textsc{Max-Go-To-The-Middle}}
\newcommand{\tauAName}{(\ensuremath{1-\tau})}

\newcommand{\fsync}{\textsc{$\mathcal{F}$sync}}

\newcommand{\innerLength}{\ensuremath{I(t)}}
\newcommand{\leftOuter}{\ensuremath{O_{\ell}(t)}}
\newcommand{\rightOuter}{\ensuremath{O_{r}(t)}}
\newcommand{\innerLengthShort}{\ensuremath{I}}
\newcommand{\leftOuterShort}{\ensuremath{O_{\ell}}}

\newcommand{\opposedConfig}{opposed configuration}

\newcommand{\maxChain}{max-chain}
\newcommand{\marchingConfig}{marching configuration}
\newcommand{\MarchingConfig}{Marching configuration}
\newcommand{\marchingChain}{marching chain}
\newcommand{\deltaUConfig}{discrete $\delta$-V-con\-fi\-gu\-ra\-tion}
\newcommand{\taudeltaUConfig}{discrete $\left(\delta, 1-\tau\right)$-V-con\-fig\-u\-ra\-tion}
\newcommand{\thetaV}{continuous $\delta$-V-con\-fi\-gu\-ra\-tion}
\newcommand{\thetaVs}{continuous $\delta$-V-con\-fi\-gu\-ra\-tions}

\newcommand{\deriv}[2]{\ensuremath{{#1}'(t)}}
\newcommand{\tauhalf}{\ensuremath{\frac{\tau}{2}}}
\newcommand{\oneminustau}{\ensuremath{\left(1-\tau\right)}}

\newcommand{\halfoneminustau}{\ensuremath{\frac{\left(1-\tau \right)}{2}}}
\newcommand{\halfoneplustau}{\ensuremath{\frac{\left(1+\tau \right)}{2}}}
\newcommand{\half}{\frac{1}{2}}
\newcommand{\quarter}{\frac{1}{4}}
\newcommand{\referenceAngle}{\ensuremath{\psi}}
\newcommand{\referenceAngleConcrete}{\ensuremath{2 \cdot \cos^{-1} \left(1-\tau\right)}}

\newcommand{\norm}[1]{\ensuremath{\|{#1}\|}}

\newcounter{PotentialIndex}
\newcounter{MatrixIndex}

\makeatletter
\def\@parfont{\bfseries}
\makeatother

\usepackage{authblk}

\title{A Discrete and Continuous Study of  the \maxForm/ Problem
\thanks{This paper is a full version of the respective paper presented at SSS 2020.
}}

\author[1]{Jannik Castenow}
\author[2]{Peter Kling}
\author[1]{Till Knollmann}
\author[1]{Friedhelm Meyer auf der Heide}
\affil[1]{Heinz Nixdorf Institute and Department of Computer Science\\
	Paderborn University,
 \{jannik.castenow, till.knollmann, fmadh\}@upb.de}
\affil[2]{Department of Informatics, Universit\"at Hamburg,  peter.kling@uni-hamburg.de}
\date{}

\begin{document}

\maketitle

\begin{abstract}
	Most existing robot formation problems seek a target formation of a certain \emph{minimal} and, thus, efficient structure.
Examples include the \Gathering/ and the \chainForm/ problem.
In this work, we study formation problems that try to reach a \emph{maximal} structure, supporting for example an efficient coverage in exploration scenarios.
A recent example is the NASA Shapeshifter project~\cite{DBLP:journals/corr/abs-2002-00515}, which describes how the robots form a relay chain along which gathered data from extraterrestrial cave explorations may be sent to a home base.

As a first step towards understanding such maximization tasks, we introduce and study the \maxForm/ problem, where $n$ robots are ordered along a winding, potentially self-intersecting chain and must form a connected, straight line of maximal length connecting its two endpoints.
We propose and analyze strategies in a discrete and in a continuous time model.
In the discrete case, we give a complete analysis if all robots are initially collinear, showing that the worst-case time to reach an $\varepsilon$-approximation is upper bounded by $\mathcal{O}(n^2 \cdot \log (n/\varepsilon))$ and lower bounded by $\Omega(n^2 \cdot~\log (1/\varepsilon))$.
If one endpoint of the chain remains stationary, this result can be extended to the non-collinear case.
If both endpoints move, we identify a family of instances whose runtime is unbounded.
For the continuous model, we give a strategy with an optimal runtime bound of $\Theta(n)$.
Avoiding an unbounded runtime similar to the discrete case relies crucially on a counter-intuitive aspect of the strategy: slowing down the endpoints while all other robots move at full speed.
Surprisingly, we can show that a similar trick does not work in the discrete model.
\end{abstract}


\section{Introduction}

Robot coordination problems deal with systems consisting of many autonomous but simple, mobile robots that try to achieve a common task.
The robots' capabilities are typically quite restricted (e.g., they have no common coordinate system or sense of direction).
Among the most well-studied tasks are \Gathering/ problems, in which robots are initially scattered and must gather at one point.
Another class of important tasks are \chainForm/ problems, where robots take the role of communication relays that, initially, form a winding chain connecting two distinguished robots.
The relays are to move such that the chain becomes straight, allowing for a more energy-efficient communication.
Applications of such chain formations can be found in the exploration of difficult terrain that restricts normal communication (e.g., cave systems)~\cite{conf/icar/NguyenPRGS03,DBLP:journals/corr/abs-2002-00515}.

Both \Gathering/ and \chainForm/ problems can be described as \emph{contracting}: starting from an initially scattered formation, they seek to reach a smaller, more efficient (communication) structure.
A natural complement to such contracting formation primitives are \emph{extension problems}.
The general idea is to spread a set of distinguished robots such that their convex hull is maximized, while maintaining a suitable connection network of simple relay robots.
We initiate the theoretical study of such problems for the case of two distinguished robots connected by a chain of relay robots.
Already this comparatively simple scenario turns out to be non-trivial to analyze.

\paragraph{Movement Model \& Time Notions}
We consider $n$ identical, oblivious, mobile robots with a limited viewing range (normalized to $1$) scattered in the Euclidean plane.
The robots form a communication chain, such that each robot has a specific predecessor and successor in distance at most $1$.
We assume no common coordinate systems.
Instead, a robot may only measure its relative position (distances and angles) to its two neighbors.
We seek a simple, deterministic\footnote{%
	Determinism implies that from certain, very symmetrical system states, robots won't be able to form a maximum length chain (e.g., when all robots start in the same position).
	This can be resolved with a very limited and small amount of randomness.
	(e.g., having the outer robots move in a random direction in such a situation).
} movement strategy that, when executed simultaneously by all robots, causes them to converge towards a straight chain of (maximal) length $n-1$.
This movement strategy takes the relative positions of the (at most) two neighboring robots and specifies where the robot moves next.
It is crucial that the distance between two neighboring robots never exceeds $1$, since otherwise we cannot guarantee that the (oblivious) robots will be able to reconnect the chain.

We refer to this as the \maxForm/ problem and study it in two different time models, the classical (synchronous) \emph{\textnormal{\LCMlong/} (\LCM/)} model and the \emph{continuous time} model.
In the \LCM/ model, time is divided into discrete \emph{rounds} in which all robots simultaneously perform a cycle of a \look/, a \compute/, and a \move/ operation.
During the \look/ operation, each robot takes a \emph{snapshot} of its neighbors' current relative positions.
Afterward, all robots start the \compute/ operation, during which they use their snapshot to compute a \emph{target point}.
Finally, all robots perform the \move/ operation by moving to the target point.
Together with our simple type of (oblivious and communication-less) robots, this is also known as the \Oblot/ model \cite{series/lncs/11340}.

The above described model is inherently discrete, which severely limits the accuracy of information on which movements are based.
The situation observed by a robot at the beginning of a round might be very different from the end of the round, when all other robots performed their movement.
This effect can be compensated, e.g., by limiting how far a robot may move towards its target point during a round.~\cite{conf/antsw/GordonWB04} considered such a model and studied how it evolves in the limit, such that robots move an infinitesimal distance per round.
This gives rise to the \emph{continuous time} model.
Here, each robot perpetually measures its neighbors' positions and, at the same time, adjusts the target point towards which it moves.
This model exhibits fundamentally different properties, as was already experimentally observed in~\cite{conf/antsw/GordonWB04} and later analytically proven in~\cite{journals/topc/DegenerKKH15} (see our detailed discussion of related work).

While the continuous model is certainly idealized, it also abstracts away the \enquote{loss of discretization} and allows one to focus on the complexity of the formation problem.
In a sense, it showcases the best possible improvement one can hope for when approaching \LCM/ cycles of length zero in practical implementations.

\paragraph{Related Work}
The following overview focuses on robot formation strategies with known runtime bounds.
In particular, we do not cover semi- or asynchronous variants of the \LCM/ model, in which the robots' \LCM/ cycles are not necessarily synchronized.
In such systems, already achieving a task like \Gathering/ may be impossible~\cite{DBLP:conf/algosensors/DieudonneP09} or requires additional robot properties~\cite{DBLP:journals/tcs/FlocchiniPSW05,DBLP:conf/sss/PoudelS17}.
The synchronous setting allows us to concentrate on the runtime analysis and to better compare the discrete and continuous models.
See~\cite{series/lncs/11340} for a quite complete and very recent survey on robot coordination problems.

The \Gathering/ problem has been considered in both the discrete and continuous setting.
Here, there is no predecessor/successor relation between the robots, and the snapshot from the \look/ operation contains all robot positions within viewing range.
A natural strategy is to move towards the center of the smallest enclosing circle spanning all robots in viewing range.
In the discrete setting, \cite{DBLP:journals/trob/AndoOSY99} proved that this strategy gathers all robots in finite time; a runtime bound of $\Theta(n^2)$ was proven later in \cite{DBLP:conf/spaa/DegenerKLHPW11}.
Up to now, this strategy achieves the  asymptotically fastest (and conjectured optimal) runtime in this model.
Taking a look at the continuous setting yields a very different situation:
\cite{conf/antsw/GordonWB04} proposed a simple, continuous strategy, in which robots try and decide locally whether they are at a vertex of the global convex hull formed by all robots.
If a robot concludes that it is at such a vertex, it moves along the angle bisector towards the inside of the (supposed) convex hull.
This strategy was shown to gather all robots in finite time.
Later, \cite{DBLP:conf/spaa/KempkesKH12} proved that the strategy's worst-case runtime is $\Theta(n)$; a considerable improvement about the $\Theta(n^2)$ bound for discrete \Gathering/.
For an overview over continuous strategies for \Gathering/, see \cite{Kling0219}.

The \chainForm/ problem was introduced and analyzed by \cite{conf/ifip10/DyniaKLH06} in the discrete setting.
The authors proposed the natural \gtmlong/ (\gtm/) strategy, in which each robot moves towards the midpoint between its two neighbors.
It is proven that \gtm/ requires $\mathcal{O}\bigl(n^2 \cdot \log(n/\epsilon)\bigr)$ rounds to reach an $\epsilon$-approximation (w.r.t.~the length) of the straight chain between the base stations.
\cite{conf/spaa/KlingH11} gave an almost matching lower bound of $\Omega\bigl(n^2 \cdot \log(1/\epsilon)\bigr)$ and generalized these bounds to a class of (linear) strategies related to \gtm/.
Note that while there are some discrete \chainForm/ strategies, specifically~\cite{journals/tcs/KutylowskiH09}, that achieve a better (linear) asymptotic runtime, such strategies are known only for relaxed models and goals (e.g., reaching only a $\Theta(1)$-approximation).
The continuous setting was analyzed by \cite{journals/topc/DegenerKKH15}, who suggested the \moblong/ (\mob/) strategy (robots move along the angle bisector formed by their two neighbors) and proved a runtime of $\Theta(n)$.
Similar to the \Gathering/ problem, we see a linear improvement when going from the discrete to the continuous setting.

Scenarios related to the idea of extension problems have been considered in other settings (like on discrete graphs) under the name uniform scattering or deployment~\cite{DBLP:conf/ipps/BarriereFBS09,DBLP:journals/jpdc/ShibataMOKM18}.
The general problem of forming a line in a distributed system has been studied in many different contexts, see e.g.\ \cite{conf/dna/GmyrHKKRSS18,conf/icinfa/JiangW018,conf/icar/NguyenPRGS03}.
While the presented theoretical models are certainly idealized (ignoring, e.g., collisions of physical robots), such algorithms can be adapted for practical systems~\cite{journals/jfr/YunAA97}.

\paragraph{Our Contribution}
We adapt the known (contracting) \chainForm/ strategies \gtm/ (discrete setting) and \mob/ (continuous setting) such that they still straighten the chain but, at the same time, keep extending its length.
The basic idea is to let inner robots perform the contracting strategy while the two outer robots extend the chain by moving away from their respective neighbor.
While this seems to be a small modification of the contracting strategies on a conceptual level, we identify a much more complex behavior of the robots caused by the extension part.
This also affects the analysis -- we use several different techniques: among others, we make use of discrete Fourier transforms, the mixing time of Markov chains and the stability theory of dynamical systems.

\Cref{section:discrete} considers the discrete setting, for which we distinguish the one-dimension\-al case (all robots are initially collinear) and the general two-dimension\-al case.
In the one-dimensional case, we already see that very symmetric configurations are problematic for any (deterministic) strategy.
This is obvious for the trivial configuration (all robots start in the same spot).
But also from less contrived starting positions (e.g., when the initial chain is symmetrical around the origin), any deterministic strategy results in a non-maximal chain (that potentially keeps moving) (see \cref{thm:cannotreachmax}).
Still, in the case of our proposed \maxgtm/ strategy, we can show:
\begin{theorem}%
	\label{thm:introduction:main:discr:onedim}
	Under the \textnormal{\maxgtm/} strategy on the line, the robot movement reaches in time  $\Omega\bigl(n^2 \cdot \log(1/\epsilon)\bigr)$ and $\mathcal{O}\bigl(n^2 \cdot \log(n/\epsilon)\bigr)$ an $\epsilon$-approximation of:
	a stationary, \maxChain{} of length $n-1$, if initially the outer robots move in different directions or
	a chain of non-maximal length moving at speed $1/n$ (\emph{marching chain}), if initially the outer robots move in the same direction.
\end{theorem}

While this gives a pretty complete picture of the one-dimensional case, the two-dimensional case exhibits a much more complex behavior.
We can still prove convergence in finite time and derive a lower bound (which now depends also on the outer robots' initial distance) but an upper bound remains elusive.
\begin{theorem}%
	\label{thm:introduction:main:discr:twodim}
	Under the \textnormal{\maxgtm/} strategy, the robot movement reaches an $\varepsilon$-approxima\-tion either of the \maxChain{} or of a one-dimensional marching chain.
	There are configurations for which this takes $\Omega(n^2 \cdot \log(1/\delta))$ rounds, where $\delta$ denotes the initial distance between the outer robots.
\end{theorem}
Interestingly, however, fixing the position of one of the two outer robots enables us to employ tools from Markov Chain theory (as used in previous results~\cite{conf/spaa/KlingH11}), yielding again the same almost tight runtime bound as in the one-dimensional case (see \cref{theorem:2-dimensions:one-stationary-upper-bound}).
Given this and some simple experimental evaluations, we conjecture that the lower worst-case bound stated in \cref{thm:introduction:main:discr:twodim} is tight.

\Cref{section:continuous} considers the continuous setting.
As in the discrete setting, very symmetric configurations again lead to unavoidable problems for deterministic strategies.
Moreover, a na\"ive translation of the \mob/ strategy results in the same dependency on the outer robots' initial distance $\delta$.
However, the continuous model allows for an interesting tweak which, as we show in \cref{section:speedOfOuterRobots}, cannot be done in the discrete model.
Namely, it turns out that decreasing the speed of outer robots by a small constant $\tau$ gets rid of the dependency on $\delta$ and yields an optimal, linear runtime bound.
As a byproduct, this also causes symmetrical initial positions to collapse to a single point instead of becoming a marching chain.
Summarized, we get the following result for the continuous setting:
\begin{theorem}%
	\label{thm:introduction:main:cont}
	\textnormal{\maxmob/} reaches in worst-case optimal time $\Theta(n)$
	\begin{enumerate*}[label=, afterlabel=]
		\item a stationary, maximum chain of length $n-1$ or
		\item the chain collapses to a single point.
	\end{enumerate*}
\end{theorem}

Our results show that the idealized continuous model yields again a linear speed-up for the \maxForm/ problem, similar as for contracting robot formation problems.
The major open problem is to find an upper runtime bound for \maxgtm/ in the discrete setting where both endpoints move.
Moreover, while very symmetrical initial configurations pose a problem for deterministic algorithms, both, experiments with a simple, custom simulator and looking at our processes from the perspective of dynamical systems~\cite{robinson2012introduction} suggest that such configurations are few and unstable.
Thus, minor, random perturbations usually yield a configuration in which the robots reach the desired maximal chain.
We analyze this observation formally by proving that the marching chain is an unstable fixed point of the related dynamical system.
We discuss this in more detail towards the end of \cref{section:discrete}.
Due to space constraints, all proofs have been deferred to the appendix.
\Cref{section:appendixProof} contains the proofs of  \Cref{section:discrete}.
The analysis of our continuous algorithm presented in \Cref{section:continuous} can be found in \Cref{section:contTwoDim} and \Cref{section:Influence-of-outer-robots-extended} contains an extended discussion and the missing proofs of \Cref{section:speedOfOuterRobots}.

\section{Model and Problem Description}%
\label{sec:model_and_preliminaries}
We follow the robot model of the \chainForm/ problem \cite{conf/sirocco/CohenP06,journals/topc/DegenerKKH15,conf/ifip10/DyniaKLH06,DBLP:conf/spaa/KempkesKH12,conf/spaa/KlingH11}:
We consider $n$ robots, $r_1, \dots , r_{n}$ that are connected in a chain topology positioned in the Euclidean plane.
The robots $r_1$ and $r_{n}$ are denoted as \emph{outer robots} and all other robots are \emph{inner robots}.
In the chain topology each inner robot $r_i$ can distinguish its two neighbors $r_{i-1}$ and $r_{i+1}$ while the robots do not have a common
understanding of left and right.
The outer robots have only a single neighbor: the neighbor of $r_1$ is $r_2$ and $r_n$'s neighbor is $r_{n-1}$.
Based on their neighborhoods, robots can detect whether they are an inner or an outer robot.
Each robot has a uniform viewing range of one.
Apart from their direct neighbors, robots cannot see any other robot that might be present in their viewing range.
In the initial configuration at time $t_0$, we assume that the chain topology is \emph{connected}, i.e.\ the distance between a robot and its neighbors is less than or equal to one.
The position of $r_i$ at time $t$ is denoted by $p_i(t) \in \mathbb{R}^2$ and for all $2 \leq i \leq n$, the vector $w_i(t) := p_i(t) - p_{i-1}(t)$ is the vector pointing from robot $r_{i-1}$ to robot $r_i$ at time $t$.
Starting at robot $r_1$, a \emph{configuration} of robots at time $t$ can be written as $w(t) := \left(w_2(t), w_3(t), \dots w_n(t)\right)^T$.
The \emph{length} of a configuration at time $t$ is denoted by $L(t) := \sum_{i=2}^{n} \norm{w_i(t)}$, where $\norm{w_i(t)}$ denotes the Euclidean norm of vector $w_i(t)$.
For a vector $w_i(t)$, we denote the normalized vector $\frac{1}{\norm{w_i(t)}} w_i(t)$ by $\widehat{w}_i(t)$.
The Euclidean distance between two robots $r_i$ and $r_j$ at time $t$ is denoted by $\Delta_{i,j}(t) = \norm{p_i(t) - p_j(t)}$.

Next, we introduce a characterization of configurations that is relevant for our analyses.
In \emph{one-dimensional} configurations, the positions of all robots are collinear.
In \emph{two-dimensional} configurations, there exists a set of at least $3$ robots whose positions are not collinear.
Our analyses distinguish two special kinds of one-dimensional configurations: \emph{Opposed configurations} and \emph{\marchingConfig{}s}.
In opposed configurations, the outer robots are on different sides of their neighbors, i.e. $\widehat{w}_2(t) = \widehat{w}_{n}(t)$.
In marching configurations, the outer robots are on the same side of their neighbors, i.e. $\widehat{w}_2(t) = - \widehat{w}_{n}(t)$.
For $2 \leq i \leq n-1$, we denote by $\alpha_i(t) := \angle \left(w_i(t), w_{i+1}(t)\right) \in [0, \pi]$ the angles along the vector chain.
Our goal is to reach a configuration with $ \Delta_{1,n}(t) = n-1$.
More precisely, each vector $w_i$ should have a length of $1$ and $w_i(t) = w_{i+1}(t)$ for $2 \leq i \leq n-1$.
We call this configuration a \emph{\maxChain{}}.
We say that we have reached an $\varepsilon$-approximation of the \maxChain{} if $  \Delta_{1,n}(t) \geq (1-\varepsilon)\,(n-1)$
and $\norm{w_i(t)} > 1-\varepsilon$ for all $2 \leq i \leq n$.

We assume a very restricted robot model, namely robots having the capabilities \text{} of the \Oblot/ model with disoriented coordinate systems and limited visibility.
Thus, the robots neither have a global coordinate system nor a common compass.
A robot can only observe the position of its neighbors relative to its own.
We assume that the robots can measure distances precisely and have a common notion of unit distance.
Additionally, the robots are \emph{oblivious} and cannot rely on any information from the past.
Furthermore, the robots cannot communicate.
Throughout this work, we consider two different notions of time, the \fsync{} time model and the continuous time model.
In \fsync{} all robots operate in fully synchronous \textsc{Look-Compute-Move} (\LCM/) cycles (rounds), i.e.; robots observe their environment, compute a target point and finally move there.
The \emph{continuous time model} can be seen as a continuous variant of the \fsync{} model for an infinitesimal small movement distance for each robot per round~\cite{conf/antsw/GordonWB04}.
In this model, robots continuously observe their environment and adjust their own movement.
There is no delay between observing the environment and adjusting the movement.
At every point in time, the movement of each robot $r_i$ can be expressed by a \emph{velocity vector} $v_i(t)$ with $0 \leq \|v_i(t)\| \leq 1$, i.e.\ the maximal speed of a robot is bounded by $1$.
The function $p_i \colon \mathbb{R}_{> 0} \to \mathbb{R}^2$ is the \emph{trajectory} of $r_i$.
The trajectories $p_i$ are continuous but not necessarily differentiable because robots are able to change their speed and direction non-continuously.
However, natural movement strategies, such as the strategy presented in this paper, have (right) differentiable trajectories.
Thus, the velocity vector of a robot $v_i \colon \mathbb{R}_{> 0} \to \mathbb{R}^2$ can be seen as the (right) derivative of $p_i$ and we can write $v_i(t) = \deriv{p_i}{t}$.

\section{The Discrete Case}%
\label{section:discrete}
In this section, we describe \AName{} for the \fsync{} time model.
Intuitively, the strategy solves two tasks concurrently.
The first task is to arrange all robots on a straight line while the second task is to lengthen the chain by moving the outer robots away from each other.
For the first task, we adapt the \gtm/-strategy for \chainForm/ in which all inner robots move to the midpoint between their neighbors in every round.
For the second task, the outer robots move as far as possible away from their neighbors while keeping the chain connected.

\subsection{\maxgtmlong/ (\maxgtm/)}%
\label{section:algorithm}

\AName{} works as follows: Every inner robot moves to the midpoint between its neighbors.
The new position of an inner robot $r_i$ can thus be computed as $p_i(t+1)    =  \half p_{i-1}(t) + \half p_{i+1}(t)$.
An outer robot moves as far possible away from its neighbor by imagining a virtual robot.
At time $t$, the outer robot $r_1$ normalizes the vector $w_2(t)$, imagines a virtual robot $r_0$ positioned at $p_1(t) - \widehat{w}_2(t)$ and moves to the midpoint between $r_0$ and $r_2$.
Thus, $p_1(t+1)    =	\half p_1(t) + \half p_2(t) - \half \widehat{w}_2(t)$.
The procedure works analogously for $r_n : p_{n}(t+1)  = \half p_{n-1}(t) + \half p_{n}(t) + \half \widehat{w}_{n}(t)  $ .
Similarly, we can derive formulas for $w(t+1)$:
$w_2(t+1)   = \half \widehat{w}_2(t) + \half w_3(t) ,
w_i(t+1)   = \half w_{i-1}(t) + \half w_{i+1}(t)$ and $
w_{n}(t+1) = \half w_{n-1}(t) + \half \widehat{w}_{n}(t)$.
Simplified, we can compute $w(t+1)$ as a matrix-vector product: $w(t+1) = S(t) \cdot w(t) =  \prod_{i=0}^{t} S(i) \cdot w(0)$ with the strategy matrix $S(t)$.
See \Cref{figure:visualizationMaxGtm,figure:matrix} for a visualization of \AName{} and the strategy matrix $S(t)$.

\begin{figure}
	\subcaptionbox{A visualization of \maxgtm/.  The target point of each robot is marked by a cross.\label{figure:visualizationMaxGtm}}[.5\linewidth]{
		\centering
		\includegraphics[width=0.5\textwidth, clip=true]{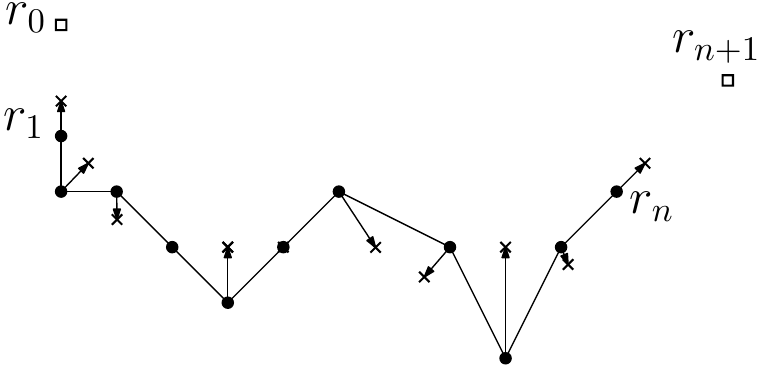}
	}
	\hfill
	\subcaptionbox{The strategy matrix $S(t)$. The upper-left and the lower-right entry depend on $w_2(t)$ and $w_n(t)$.	\label{figure:matrix}}[.45\linewidth]{
		$S(t) = $ \\
		\hfill \\
		$ \begin{bmatrix}
			\frac{1}{2 \cdot \norm{w_2(t)}} & \half  & 0      & 0      & \dots  & 0      & 0                               \\
			\half                           & 0      & \half  & 0      & \dots  & 0      & 0                               \\
			0                               & \half  & 0      & \half  & \dots  & 0      & 0                               \\
			\vdots                          & \vdots & \vdots & \vdots & \vdots & \vdots & \vdots                          \\
			0                               & 0      & 0      & 0      & \dots  & \half  & \frac{1}{2 \cdot \norm{w_n(t)}}
		\end{bmatrix}$
	}
	\caption{A visualization of \maxgtm/ and the strategy matrix $S(t)$.}
\end{figure}

\subsection{One-Dimensional Analysis} \label{section:1-dimension}

This section investigates the performance of \AName{} in a one-dimensional configuration.
One-dimensional configurations already reveal an interesting behavior of \maxgtm/: in some configurations, the strategy does not converge to a \maxChain{} but to a different structure, we will denote as the \marchingChain{}.
The two classes of configurations that play a role in this analysis are marching and opposed configurations.
We can show that \AName{} does not switch between the two classes.

\begin{restatable}{lemma}{lemmaMaxGtmSwitchOpposedMarching}\label{lemma:1-dimension:orientation-of-outer-robots}
	\textnormal{\AName{}} does not switch between opposed and marching configurations.
\end{restatable}

As a consequence of \Cref{lemma:1-dimension:orientation-of-outer-robots}, starting from a \marchingConfig{}, \AName{} does not converge to a \maxChain{}.
For some highly symmetric configurations, for instance the configuration depicted in \Cref{figure:problem-description:marching-chain}, this even cannot be obtained by any deterministic strategy because the outer robots always have the same view and, thus, always stay on the same position.

\begin{restatable}{theorem}{theoremMaxGtmMarchingSymmetry}%
	\label{thm:cannotreachmax}
	There are marching configurations that cannot be transformed into a \maxChain{} by any deterministic strategy.
\end{restatable}

For \opposedConfig{}s, we can show that \AName{} converges towards an \( \varepsilon \)-approxima\-tion of the \maxChain{}.
Define $m_i(t) = 1 - w_i(t)$.
For the analysis, we use the potential function  \( \phi_{1} (t) = \sum_{i=2}^{n} m_i(t)^2 \) as a progress measure.
Intuitively, $\phi_{1}(t)$ measures how close the configuration is to the \maxChain{}, since in the \maxChain{}, $w_i(t) = 1$ for all $i$.
For $m_1(t) = 0$ and $m_{n+1}(t) = 0$, it holds $\phi_{1}(t+1) = \sum_{i=2}^{n} m_i(t+1)^2 = \sum_{i=2}^{n}\bigl(\frac{m_{i-1}(t) + m_{i+1}(t)}{2}\bigr)^2$.
Inspired by \cite{conf/sirocco/CohenP06}, we can analyze $\phi_{1}(t)$ with the help of discrete Fourier Transforms.
Discrete Fourier Transforms are useful here, because they allow us to decouple the computations of the $m_i(t+1)$'s.
By now, we can express $m_i(t+1)$ based upon $m_{i-1}(t)$ and $m_{i+1}(t)$.
Discrete Fourier Transforms remove this dependency such that we get a single (non-recursive) formula for each $m_i(t+1)$ and can bound $\phi_{1}(t+1)$ by the slowest decreasing $m_i(t)$, resulting in an upper runtime bound of $\mathcal{O}\left(n^2 \cdot \log \left(n/\varepsilon\right)\right)$.

\begin{restatable}{theorem}{theoremFourier}\label{theorem:1-dimension:opposite-directions-upper-bound}
	Started in an \opposedConfig{}, \textnormal{\AName{}} needs at most \( \mathcal{O}(n^{2} \cdot \log \left(n/\varepsilon \right)) \) rounds to achieve an \(\varepsilon\)-ap\-prox\-i\-ma\-tion of the \maxChain{}.
\end{restatable}

The analysis of the \emph{mixing time} of a Markov Chain, allows us to prove a close lower bound of $\Omega \left(n^2 \cdot \log  \left(1/\varepsilon\right)\right)$.
Since \opposedConfig{}s remain \opposedConfig{}s, we can rewrite $w(t)$ and $S(t)$ slightly such that the resulting strategy matrix is stochastic (every row sums up to $1$).
Hence, this matrix could also be the transition matrix of a Markov Chain.
This Markov Chain has a single absorbing state.
In general, there are no mixing time bounds for absorbing Markov Chains since this type of Markov Chains does not have a unique stationary distribution.
Markov Chains with a single absorbing state, however, have a unique stationary distribution, such that some bounds for the mixing time exist.
Here, we can make use of a lower runtime bound.

\begin{restatable}{theorem}{theoremMaxGtmLowerBoundOpposed}
	\label{theorem:1-dimension:lower-bound-opposed-configs}
	There exist \opposedConfig{}s such that  \textnormal{\AName{}} needs at least \( \Omega \bigl(n^2 \cdot \log (1/\varepsilon)\bigr)\) rounds to achieve an \( \varepsilon \)-ap\-proxi\-mation of the \maxChain{}.
\end{restatable}

\MarchingConfig{}s do not converge to a \maxChain{} but have a different convergence behavior, they converge to the \emph{\marchingChain}.
It is called \marchingChain{} because the robots all together move into the same direction and never stop.
The configuration $w_M(t)$ defines the \marchingChain{}.
$w_M(t) = (1- \frac{2}{n}, 1 - \frac{4}{n}, \dots, \frac{2}{n} ,0,$
$ - \frac{2}{n}, - \frac{4}{n}, \dots, -(1- \frac{2}{n}))^T$.
\Cref{figure:problem-description:marching-chain} visualizes this \marchingChain{}.
Observe that $S(t) \cdot w_M (t)= w_M(t)$ ($w_M(t)$ is an eigenvector of $S(t)$ to the eigenvalue $1$).
In the marching chain, each robots moves distance $\frac{1}{n}$ per round.

\begin{figure}[htb]
	\centering
	\includegraphics[width=0.8\textwidth, clip=true]{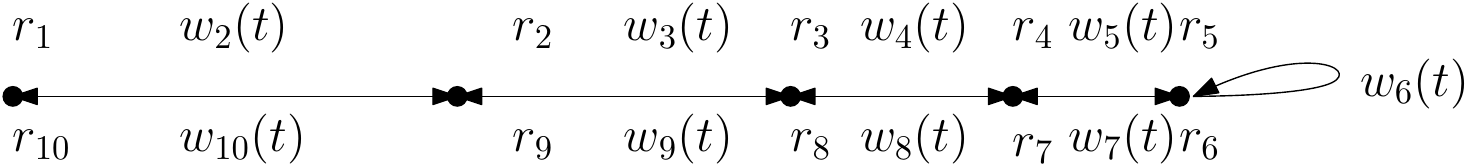}
	\caption{A marching chain for $n=10$. Every position is occupied by two robots.}
	\label{figure:problem-description:marching-chain}
\end{figure}

Starting in a \marchingConfig{}, the convergence time until all vectors only differ up to $\varepsilon$ from their corresponding vector in the \marchingChain{} is equal to the runtime bound for opposed configurations.
Here, we can again use the analysis of the mixing time of a Markov Chain for a slightly different transition matrix.
We consider the vectors $z_i(t) = p_i(t+1) - p_i(t)$ pointing from the current position of a robot to its next position.
The changes of the vectors $z_i(t)$ can be calculated via a matrix-vector product of a doubly stochastic transition matrix and the current vectors.
A doubly stochastic matrix can also be the transition matrix of an irreducible, aperiodic and reversible Markov Chain.
These Markov Chains converge to a unique stationary distribution and the runtime bounds depend on the second largest eigenvalue of the transition matrix.
By analyzing this second largest eigenvalue we can prove the following runtime bounds:

\begin{restatable}{theorem}{theoremMaxGtmMarchingBounds}\label{theorem:1-dimension:upper-bound-marching-configs}
	Given a \marchingConfig{}, \textnormal{\AName{}} needs at most \( \mathcal{O}(n^{2} \cdot \log (n/\varepsilon)) \) and at least $\Omega(n^2 \cdot \log (1/\varepsilon))$ rounds to achieve an \(\varepsilon\)-approximation of the \marchingChain{}.
\end{restatable}

\stepcounter{PotentialIndex}

\subsection{Two-Dimensional Analysis} \label{section:2-dim}
Next, we prove a convergence result for two-dimensional configurations, stating that an initial configuration either converges to the \maxChain{} or to the \marchingChain{}.
In the analysis, we again consider the vectors $z_i(t) = p_i(t+1) - p_i(t)$ that we have already seen in the upper runtime bound for marching configurations.
The potential function $\phi_{2}(t)$ is the sum of all squared lengths of $z_i(t)$'s:
$\phi_{2}(t) = \sum_{i=1}^{n} \norm{z_i(t)}^2$.
$\phi_{2}(t)$ is a monotonically decreasing function and the potential difference 	$	\phi_{2}(t) - \phi_{2}(t+1)$ can be bounded as stated by \Cref{lemma:2-dimensions:potential-difference}.
With help of this potential difference, we can conclude that two-dimensional configurations either converge to a \marchingChain{} or to the \maxChain{}.

\begin{restatable}{theorem}{theoremConvergence} \label{theorem:2-dimensions:convergence}
	Given an arbitrary connected chain in the Euclidean plane, \textnormal{\AName{}} converges either to the \marchingChain{} or to the \maxChain{}.
\end{restatable}

\stepcounter{PotentialIndex}
\stepcounter{MatrixIndex}

Interestingly, when assuming that only one of the outer robot moves while the other one remains stationary, we can prove the same upper runtime bound as for one-dimensional configurations.
The proof relies on the analysis of $\phi_2(t)$  for this case.
Again, the analysis of a transition matrix plays a role here -- since only one outer robot moves we obtain a \emph{substochastic} transition matrix where every row except of one sums up to $1$.
The last row only sums up to $1/2$ such that high powers of this matrix converge to the $0$-matrix.
A diagonalization of the transition matrix yields the following runtime bound:

\begin{restatable}{theorem}{theoremMaxGtmOneMovingUpper}\label{theorem:2-dimensions:one-stationary-upper-bound}
	In case one of the outer robots is stationary and all other robots move according to \textnormal{\AName{}}, an $\varepsilon$-approximation of the
	\maxChain{} is achieved after $\mathcal{O}(n^2 \cdot  \log (n/\varepsilon))$ rounds.
\end{restatable}

In case both outer robots move, we identify a certain class of configurations that lead to an arbitrarily high runtime based on a parameter $\delta$ which can be seen as the width of the configuration.
Before defining the configurations, we give some intuition about their construction: Applying \maxgtm/ to the configuration of robots can be interpreted as a discrete time \emph{dynamical system}.
See \cite{robinson2012introduction} for an introduction to dynamical systems.
In \Cref{section:1-dimension}, we have seen that this dynamical system has two different (classes of) \emph{fixed points}, i.e.\ a configuration that remains unchanged when applying \maxgtm/.
These fixed points are the \maxChain{} and the \marchingChain{}.
We can prove that the \marchingChain{} is an \emph{unstable} fixed point.
Unstable means that a small perturbation in the configuration results in a different behavior -- the dynamical system moves away from this fixed point.
In our case this means that a small perturbation in the \marchingChain{} leads to a configuration that converges to the \maxChain{}.
For a formal description of the relation to dynamical systems and a proof that the \marchingChain{} is an unstable fixed point, we refer the reader to \Cref{section:discreteDynamicalSystems}.
We use the property that \marchingChain{s} are unstable fixed points to define configurations which are very close to the \marchingChain{}, \emph{\deltaUConfig{}s}.
For $\delta = 0$, \deltaUConfig{}s and \marchingChain{}s coincide.
Choosing \emph{any} $\delta > 0$ changes the behavior such that the configuration converges to the \maxChain{}.
The runtime, however, can be arbitrarily high depending on $\delta$.
For a visualization of \deltaUConfig{}s, see \Cref{fig:discreteThetaV} in \Cref{section:speedOfOuterRobots}.

\begin{definition} \label{definition:2-dimensions:delta-u-configurations}
	For $n$ even, a \emph{\deltaUConfig{}} is defined by the vectors $w_i(t) := (\frac{\delta}{n-1},1-\frac{2\left(i-1\right)}{n})^T$ for \( i = 2,\dots, n \).
\end{definition}

\begin{restatable}{theorem}{deltaU} \label{theorem:2-dimensions:general-lower-bound}
	Starting in a \deltaUConfig{}, \textnormal{\AName{}} needs at least $\Omega\bigl(n^2 \cdot \max \{ \log (1/\varepsilon),$ $\log  (1/\delta) \} \bigr)$ rounds
	to achieve an $\varepsilon$-approximation of the \maxChain{}.
\end{restatable}

As a consequence, the runtime of \maxgtm/ can be arbitrarily depending on $\delta$.
Interestingly, the dependence on $\delta$ can be removed in the continuous time model by an -- at the first sight -- counter-intuitive approach: The outer robots move slower than the inner robots.
The same approach, however, does not work in \fsync{} (see \Cref{section:speedOfOuterRobots}).

\section{The Continuous Case}%
\label{section:continuous}

This section is dedicated to the  \parameterizedMaxMoveOnBisector{1}{1-\tau} strategy that transforms a connected chain into a \maxChain{} in the continuous time model.
After introducing the strategy, we continue with some preliminaries in \Cref{section:prelim} and provide an intuitive explanation of the strategy combined with a proof outline in \Cref{section:oneDimeOutline}.

\paragraph{\parameterizedMaxMoveOnBisectorLong{1}{1-\tau} (\maxmob/)}
Outer robots move with a maximal speed of $\left(1-\tau\right)$ for a constant $0 < \tau \leq 1/2$ as follows:
In case $\norm{w_2(t)} < 1$: $v_1(t) = - \left(1-\tau\right) \cdot \widehat{w}_2(t)$.
Similarly, in case $\norm{w_n(t)} < 1$: $v_n(t) =  \left(1-\tau\right) \cdot \widehat{w}_n(t)$.
In other words, outer robots move with a speed of $\left(1- \tau\right)$ away from their direct neighbors.
Otherwise, provided $\norm{w_2(t)} = 1$ ($\norm{w_n(t)} = 1$ respectively), an outer robot adjusts its own speed and tries to stay in distance $1$ to its neighbor while moving with a maximal speed of $1 - \tau$.
An inner robot $r_i$ with $0 <\alpha_{i}(t) < \pi$ moves only if at least one of the following three conditions holds: $\|w_{i}(t) \| = 1$, $\|w_{i+1}(t) \| = 1$ or $\alpha_i(t) < \referenceAngle{} $ for $\referenceAngle{} := \referenceAngleConcrete{}$.
Otherwise an inner robot does not move at all.
In case one of the conditions holds, an inner robot moves with speed $1$ along the angle bisector formed by the vectors pointing to its neighbors.
As soon as the position of the robot and the positions of its neighbors are collinear it continues to move with speed $1$ towards the midpoint between its neighbors while ensuring to stay collinear.
Once it has reached the midpoint it adjust its own speed to stay on the midpoint.
See \Cref{figure:strategyDescription} for a visualization.

\subsection{Preliminaries} \label{section:prelim}

For both outer robots we determine the index of the first robot that is not collinear with its neighbors and the outer ro\-bot.

\begin{definition} \label{definition:leftAndRightIndices}
	$\ell(t)$ is the index, s.t.\ for all $2 < j \leq \ell(t)$ either $w_j(t) = (0,0)$ or $\widehat{w}_j(t) = \widehat{w}_2(t), w_{\ell(t)+1}(t) \neq (0,0)$ and $\widehat{w}_{\ell(t)+1}(t) \neq \widehat{w}_2(t)$.
	Similarly, define $r(t)$ to be the index such that for all $r(t) < j < n$ either $w_j(t) = (0,0)$ or $\widehat{w}_j(t) = \widehat{w}_{n}(t)$ and $w_{r(t)}(t) \neq (0,0)$ and $\widehat{w}_{r(t)}(t) \neq \widehat{w}_{n}(t)$.
	In case there is no such an index  define $\ell(t) = r(t) = 0$.
	$\alpha_{\ell(t)}(t)$ and $\alpha_{r(t)}(t)$ are denoted as \emph{outer angles}.
\end{definition}
We omit the time parameter $t$ when it is clear from the context, e.g., we write $\alpha_{\ell}(t)$ instead of $\alpha_{\ell(t)}(t)$.
In addition to the indices $\ell(t)$ and $r(t)$, we define the last indices of robots (starting to count at $\ell(t)$ and $r(t)$) that are collinear with their neighbors and the corresponding robot with index $\ell(t)$ or $r(t)$.

\begin{definition}
	Let $\ell^{+}(t)$ be the smallest index larger than $\ell(t)$ such that $\alpha_{\ell^{+}}(t) < \pi$.
	Similarly let $r^{+}(t)$ be the largest index less than $r(t)$ such that $\alpha_{r^{+}}(t) < \pi$.
\end{definition}

\begin{figure}[hbt]
	\begin{minipage}[hbt]{0.48\textwidth}
		\centering
		\vspace*{0.4cm}
		\includegraphics[width=\textwidth, clip=true]{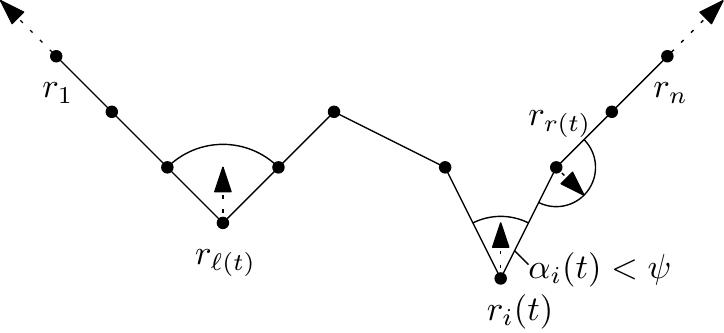}
		\caption{A chain visualizing the movements of \parameterizedMaxMoveOnBisector{1}{1-\tau}. The velocity vectors are depicted by dashed arrows.}
		\label{figure:strategyDescription}
	\end{minipage}
	\hfill
	\begin{minipage}[hbt]{0.48\textwidth}

		\centering
		\vspace*{0.63cm}
		\includegraphics[width=\textwidth]{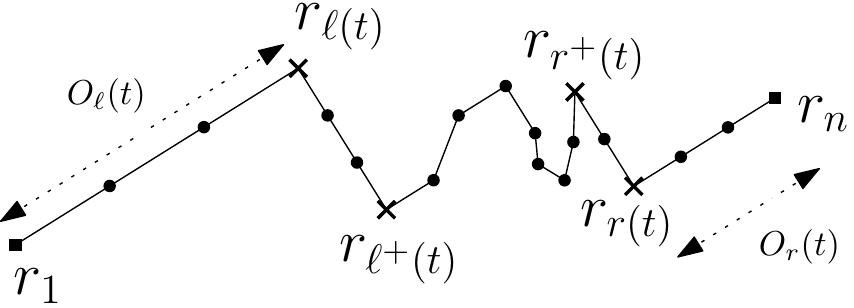}
		\caption{%
			A visualization of  $\ell(t), \ell^{+}(t), r(t)$, $r^{+}(t)$, $\leftOuter{}$ and $\rightOuter{}$.
		}	\label{figure:definition}

	\end{minipage}
\end{figure}

\begin{definition}
	The \emph{left outer length} $\leftOuter{} := \sum_{i=2}^{\ell(t)} \norm{w_i(t)}$ and the  \emph{right outer length} $\rightOuter{} := \sum_{i=r(t)+1}^{n} \norm{w_i(t)}$.
	The maximal values of the left and right outer length are denoted by $\gamma_\ell(t) := \ell(t)-1$ and $\gamma_r(t) := n-r(t)$.
	Additionally, the \emph{inner length} is defined as $\innerLength{} := \sum_{i=\ell(t)+1}^{r(t)} \norm{w_i(t)}$. See \Cref{figure:definition}.
\end{definition}

\subsection{Intuition \& Proof Outline} \label{section:oneDimeOutline}
We only give the high level idea here.
For the complete analysis, we refer to \Cref{section:contTwoDim}.
The main idea of  \parameterizedMaxMoveOnBisector{1}{1-\tau} is to flatten and stretch the chain starting at the outer robots towards the inside of the chain.
At first, $\norm{w_2(t)} = 1$ and $\norm{w_n(t)} =1$ is ensured, afterwards the angles $\alpha_{2}(t)$ and $\alpha_{n-1}(t)$ should reach a size of $\pi$ and so on until finally all vectors have length $1$ and all angles have a size of $\pi$.
\Cref{figure:coreIdea} visualizes this core idea.
To achieve this behavior, one of the two cases in which an inner robot $r_i$ moves demands either $\norm{w_{i}(t)} =1$ or $\norm{w_{i+1}(t)} =1$, because locally it can assume that it is already located on the straight line to one outer robot and all vectors into the direction of the outer robot have a length of $1$.
In addition, an inner robot $r_i$ moves if $\alpha_{i}(t) < \referenceAngle{}$.
In \Cref{section:speedOfOuterRobots} we prove that this property is crucial for the linear runtime of the strategy by introducing configurations that have a high runtime that not only depends on the number of robots in case the property is ignored.
To express the behavior of flattening and stretching the chain starting at the outer robots towards the inside of the chain, we have introduced the indices $\ell(t)$ and $r(t)$.
For each of the two sets of robots $r_1, \dots, r_{\ell(t)}$ and $r_{r(t)}, \dots, r_n$ it always holds that these robots continue to stay collinear for the rest of the execution.

For the analysis, the outer angles $\alpha_{\ell}(t)$ and $\alpha_r(t)$ play an important role.
Our analysis divides the possible sizes into three intervals, $[0, \referenceAngle{}), [\referenceAngle, \frac{3}{4} \pi]$ and $(\frac{3}{4} \pi, \pi]$.
In each interval, \maxmob/ has a certain behavior.
For an outer angle  $ \alpha_i(t) \in [0,\psi) (i \in \{\ell(t), r(t)\})$ two properties hold:
$I(t)$ decreases with speed at least $1-\tau$  and the corresponding outer length decreases since the outer robots move with speed at most $1-\tau$.
As $I(t)$ decreases with a constant speed of at least $1-\tau$, the total time in such a case is upper bounded by $\mathcal{O}\left(n\right)$.
Given $\alpha_i(t) \in [\psi, \frac{3}{4}\pi]$, the strategy is designed such that $r_i$ only moves if $O_i(t) = \gamma_i(t)$.
Thus, as long as $O_i(t) < \gamma_i(t)$, $O_i(t)$ increases with speed $1-\tau$.
As soon as $O_{i}(t) = \gamma_i(t)$ holds, the robot $r_i$ starts moving along its bisector.
This movement causes a decrease of $I(t)$ with speed at least $\cos \left(\frac{3}{8}\pi\right)$ while the length of $O_i(t)$ does not change.
Since, $I(t)$ decreases with constant speed, this case can hold for time at most $\mathcal{O}(n)$.
For the last interval, $\alpha_{i}(t) \in (\frac{3}{4}, \pi]$ we use a different progress measure since large angles cause a very slow decrease of $I(t)$ which cannot be bounded by a constant anymore.
Therefore, we consider the height $H_i(t)$.
Assume $i = \ell(t)$, then $H_{\ell}(t)$ denotes the distance between $r_{\ell(t)}$ and the line segment connecting $r_1$ and $r_{\ell^{+}(t)}$.
Intuitively, if we consider the line segment connecting $r_1$ and $r_{\ell^{+}(t)}$ as a line parallel to the $x$-axis, the robot $r_{\ell}$ moves with a velocity vector that has a small angle to the $y$-axis towards this line segment.
Thus, $H_i(t)$ decreases with constant speed.
All in all, we can prove that the total time of outer angles in any of these intervals is bounded by $\mathcal{O}\left(n\right)$ such that finally $\ell(t) = r(t)$ holds.
Lastly, we analyze the case $\ell(t) = r(t)$ and prove a linear runtime for the strategy (\Cref{theorem:mainTheorem}).

\begin{figure}[hbt]
	\centering
	\includegraphics[width=1\textwidth, clip=true]{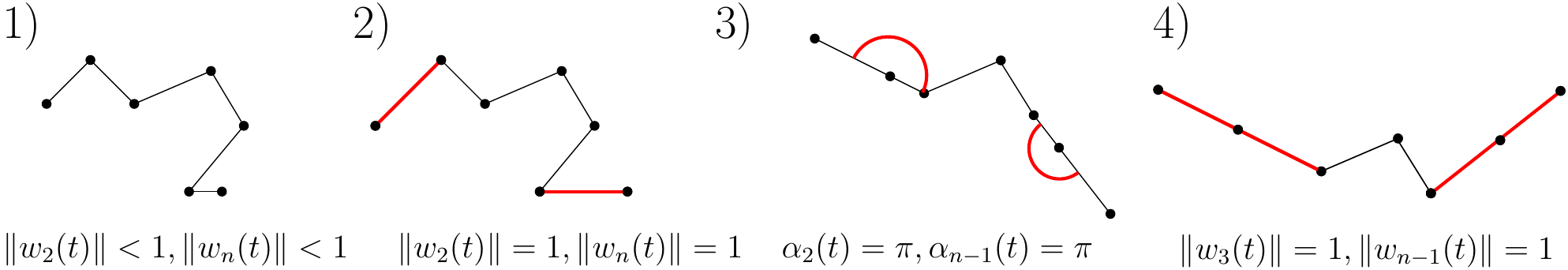}
	\caption{A visualization of the core idea of \parameterizedMaxMoveOnBisector{1}{1-\tau}. $1)$ depicts an initial configuration. $2)$ visualizes the configuration after stretching $w_2(t)$ and $w_n(t)$. In $3)$, $\alpha_{2}(t) = \pi$ and $\alpha_{n-1}(t) = \pi$.  In $4)$ $r_1$ and $r_2$ as well as $r_{n-1}$ and $r_n$ have moved such that $\norm{w_2(t)} = \norm{w_3(t)} = \norm{w_{n-1}(t)} = \norm{w_n(t)} = 1$. }
	\label{figure:coreIdea}
\end{figure}

\begin{restatable}{theorem}{theoremMainTheorem} \label{theorem:mainTheorem}
	Starting \textnormal{\parameterizedMaxMoveOnBisector{1}{1- \tau}} in a two-dimensional configuration,
	the initial chain is either transformed into a straight line of length $n-1$ or all robots are located at the same position after time $2 \left(n-3\right) \bigl(\frac{1}{1-\tau} + \frac{1}{\sqrt{2-\sqrt{2}}} + 10\bigr) + 3n \bigl(\frac{1}{\tau} + \frac{1}{1-\tau}\bigr)$
\end{restatable}

There might be two-dimensional configurations in which the chain contracts to a single point instead of reaching the \maxChain{}.
Our simulations support the following conjecture.

\begin{restatable}{conjecture}{conjectureLebesgueContinuousTwo} \label{conjecture:lebesgueContinuous2}
	The set of initial two-dimensional configurations that result in a configuration where all robots are located on the same position when applying \textnormal{\parameterizedMaxMoveOnBisector{1}{1-\tau}} has Lebesgue measure~$0$.
\end{restatable}

\section{On the Speed of the outer Robots} \label{section:speedOfOuterRobots}

We close by a brief discussion on the influence of the speed of the outer robots.
An elaboration on it can be found in \Cref{section:Influence-of-outer-robots-extended}.
It turns out there exists a special class of configurations -- called \thetaV{}s -- parameterized in the initial distance of the two outer robots \( \delta \) for which \parameterizedMaxMoveOnBisector{1}{1-\tau} needs a runtime independent of \( \delta \), see \cref{lemma:continuousThetaV}.

\begin{restatable}{theorem}{lemmaMaxMobContinuousThetaV} \label{lemma:continuousThetaV}
	Starting in a \thetaV{}, \textnormal{\parameterizedMaxMoveOnBisector{1}{1-\tau}} needs time at most $n \cdot \bigl(\frac{1}{\tau} + \frac{1}{1-\tau}\bigr) $ to transform the configuration into a \maxChain{}.
\end{restatable}

One might suspect that an algorithm in which the outer robots move at full speed stretches the chain faster.
Interestingly, this is not true!
For such an algorithm -- let it be called \naivemaxmob/ -- we can show that the runtime for \thetaV{}s is lower bounded dependent on \( \delta \).

\begin{restatable}{theorem}{theoremOneOneMaxMob} \label{theorem:thetaTrianglesRuntime}
	\textnormal{\naivemaxmob/}  transforms a \thetaV{} into a \maxChain{} in time $\Omega \left(n \cdot \log \left(1/\delta\right)\right)$.
\end{restatable}

So, slowing down the outer robots actually allows us to achieve a runtime independent of the initial configuration in the continuous case.
Can we apply the same idea in the discrete model?
Unfortunately not!
Consider the algorithm $(1-\tau)$-\textnormal{\AName{}} that behaves as \AName{} except that the outer robots move always by a distance of \( (1-\tau) \) times the distance they would in \AName{}.
Similar to \deltaUConfig{}s for \AName{}, there exists the class of \taudeltaUConfig{}s in which $(1-\tau)$-\textnormal{\AName{}} has a runtime depending on $\delta$.

\begin{restatable}{theorem}{theoremOneMinusTauGtmLower} \label{theorem:2-dimensions:tau-lower-bound}
	Starting in a \taudeltaUConfig{}, $(1-\tau)$-\textnormal{\AName{}} needs at least $\Omega\bigl(n^2 \cdot \max \{ \log (1/\varepsilon), \log  (1/\delta)\} \bigr)$ rounds
	to a\-chieve an $\varepsilon$-approximation of the \maxChain{}.
\end{restatable}

\bibliographystyle{splncs04}
\bibliography{DBLPelizer,stretching}

\newpage
\appendix

\section{Omitted Proofs for \maxgtm/} \label{section:appendixProof}

\subsection{Preliminary Theorems} \label{section:discretePrelim}

\begin{theorem}[\cite{levin2017markov,2013arXiv1310.8021J}] 	\label{theorem:1-dimension:lower-bound-mixing-time}
	Let $P$ be the transition matrix of a reversible, irreducible and aperiodic Markov Chain over a state space of size $n$.
	Furthermore, let $\pi_{min}(P)$ be the smallest entry of its stationary distribution $\pi(P)$ and $\lambda_{2}$ the second largest absolute eigenvalue of $P$.  Then, it holds
	$
	\left(\frac{1}{1-\lambda_{2}} -1 \right) \cdot \frac{1}{\log (2 \varepsilon)} \leq t_{mix}(\varepsilon) \leq 	\left(\frac{1}{1-\lambda_{2}} -1 \right) \cdot \ln \left(\frac{1}{\varepsilon \cdot \pi_{min}(P)}\right)
	$.
	The lower bound even holds in case the Markov Chain is not irreducible and aperiodic, it must only converge to a unique stationary distribution.
\end{theorem}

\begin{theorem}[Lemma 2 in \cite{conf/ifip10/DyniaKLH06}] \label{theorem:preliminaries:powers-of-substochastic-matrices}
	For any irreducible, symmetric, substochastic $n \times n$ matrix P with pairwise distinct eigenvalues and any $i, j$ we have
	$P^k[i,j] \leq n \cdot \alpha \cdot \beta ^k$
	where $\beta$ is the largest absolute eigenvalue of $P$ and
	$\alpha = \max_{i,j,i',j'} |x_j[i] \cdot x_{j'}[i']|$ with $x_j$ denoting the
	$j$-th eigenvector of $P$.
\end{theorem}

\begin{restatable}{theorem}{theoremMaxGtmPrelimDiscreteFourier}
	\label{theorem:preliminaries:discrete-sine-transformations}
	Let \( t \) be a time step and \( z_{1}(t), \dots, z_{n+1}(t) \) be real values with \( z_{1}(t) = z_{n+1}(t) = 0 \).
	Let \( \phi \) be a potential function with \( \textbf{1)} \phi (t) = \sum_{i=1}^{n+1} z_{i}(t)^2 \) and
	\( \textbf{2)}\, \phi (t+1) = \sum_{i=2}^{n} \left(\frac{z_{i-1}(t) + z_{i+1}(t)}{2}\right)^2 \).
	After \( t \in \mathcal{O} (n^2 \log (x)) \) rounds for \( x\in \mathbb{R}, \,  \phi (t) \leq \frac{1}{x^{2}} \phi (0). \)

\end{restatable}

\begin{proof}
	The potential function \( \phi(t) \) is due to the properties $1)$ and $2)$  nearly identical to the function \( \psi(t) \) of \cite{conf/sirocco/CohenP06}.
	In particular, this allows us to apply the same discrete sine transformations and the proof of Theorem 2.2 of \cite{conf/sirocco/CohenP06} works analogously up to Equation 5 resulting in
	$\phi (t+1) \leq \cos^{2} \left(\frac{\pi}{n+1}\right)\,\phi (t) \leq \cos^{2\,t} \left(\frac{\pi}{n+1}\right)\,\phi (0)$ for any time step \( t \).
	Note that the difference in the denominator comes from the fact that we consider \( n \) moving robots in contrast to the \( N - 2 \) moving robots assumed in \cite{conf/sirocco/CohenP06}.
	Next, using standard methods and Taylor's theorem, we can derive that for \( y \in \mathbb{R} \) and \( 0 \leq y \leq \frac{\pi}{2} \):
	$\cos (y) \leq \left(1-\frac{y^{2}}{4}\right)^2$.
	Setting \( t = \left\lceil \frac{(n+1)^2}{\pi^2}\,\ln (x^{2}) \right\rceil \in \mathcal{O}(n^2 \log (x^{2})) \) then yields
	$\phi (t) \leq \frac{1}{x^{2}} \phi (0).$

\end{proof}

\subsection{One-Dimensional Analysis} \label{appendix:discreteOneDim}

For the analysis, we assume that the robots are distributed on the x-axis of a two-dimensional Cartesian coordinate system (not known to the robots).
We divide the analysis into two parts based on the initial configuration of the robots.
For \opposedConfig{}s, we assume w.l.o.g.\ that $r_1$ moves to the left (into negative direction) and $r_{n}$ moves to the right (into positive direction).
Define $d$ to be $-1$ for \opposedConfig{}s and $1$ for \marchingConfig{}s.
Then, the formulas for $p_i(t+1)$ simplify for one dimension to
$p_1(t+1)    = \half p_1(t) + \half p_2(t) + \half d$ and  $p_{n}(t+1) = \half p_{n-1}(t) + \half p_{n}(t)+ \half$.
Similarly, the vectors $w_i(t)$ are one-dimensional vectors here and the equations for $w_2(t+1)$ and $w_{n}(t+1)$ can be expressed as  $w_2(t+1)    = \half w_3(t)  -  \half d$ and
$w_{n}(t+1) = \half w_{n-1}(t) + \half$.

\lemmaMaxGtmSwitchOpposedMarching*

\begin{proof}
	Consider an outer robot, for example \( r_{1} \).
	Assume w.l.o.g. that initially \( p_{1}(t) < p_{2}(t) \) and that \( r_{1} \) is passed by \( r_{2} \) at time \( t + 1 \).
	Then,
	$
	p_{1}(t+1) > p_{2}(t+1)
	\Leftrightarrow \frac{1}{2} p_{1}(t) + \frac{1}{2} p_{2}(t) - \frac{1}{2} > \frac{1}{2} p_{1}(t) + \frac{1}{2} p_{3}(t)  \Leftrightarrow p_{2}(t) - 1 > p_{3}(t).
	$
	This implies that the distance between the two neighbored robots \( r_{2} \) and \( r_{3} \) was greater than 1 at time \( t \) which cannot be.
	Using similar arguments for the other cases, \Cref{lemma:1-dimension:orientation-of-outer-robots} follows.
\end{proof}

\theoremMaxGtmMarchingSymmetry*

\begin{proof}
	Assume that there is a deterministic strategy \( s \) that transforms an initial marching configuration into a \maxChain{}.
	Then \( s \) must switch the marching configuration to an opposed configuration.
	Consider an initial marching configuration \( w(0) \) that is symmetric in the sense that the local situation of any robot \( r_{i} \) is equivalent to the local situation of robot \( r_{n+1-i} \) for each \( 1 \leq i \leq n  \).
	Assume that \( s \) transforms \( w(0) \) into an \opposedConfig{} at time step \( t \).
	Thus, \( r_{1} \) changes its direction at time \( t \).
	Since the configuration is symmetric, \( s \) is deterministic, and the robots are anonymous, \( r_{n} \) changes its direction at time \( t \) as well.
	Then \( w(t+1) \) is still a marching configuration which poses a contradiction.
\end{proof}

\theoremFourier*

\stepcounter{PotentialIndex}
\begin{proof}
	We use the potential \( \phi_{1} (t) = \sum_{i=2}^{n}(1- \norm{w_i(t)})^2 \) for a time step \( t \).
	First, observe that for \( m_{i}(t) = 1 - \norm{w_i(t)} \), it holds that
	$\phi_{1}(t) = \sum_{i=2}^{n} m_{i}(t)^2$ and
	$\phi_{1}(t+1) = \sum_{i=2}^{n} \left(\frac{m_{i-1}(t) + m_{i+1}(t)}{2}\right)^2.$
	Also, since the virtual robots \( r_{0} \) and \( r_{n+1} \) are always in distance 1 to their neighbors, \( m_{1}(t) = m_{n+1}(t) = 0 \).
	\( \phi_{1}(t) \) fulfills properties $1)$ and $2)$ of \Cref{theorem:preliminaries:discrete-sine-transformations} which yields
	$\phi_{1} (t)
	\leq \left(\frac{\varepsilon}{n-1}\right)^2  \cdot \phi_{1}(0)
	\leq \frac{\varepsilon^{2}}{n-1}$ for a \( t \in \mathcal{O}(n^2 \log (n/\varepsilon)) \).
	The last inequality holds, because \( \phi_{1}(0) \) is upper bounded by \( n-1 \).
	Obviously, there cannot be any vector $w_i(t)$ fulfilling $\|w_i(t)\| < 1 - \varepsilon$ since $\Phi_1(t) \leq \frac{\varepsilon^2}{n-1}$ holds.
	Thus, for each individual vector $w_i$ it holds $\|w_i\| > 1-\varepsilon$ and thus, a $\varepsilon$-approximation of the max-chain is achieved.
\end{proof}

\theoremMaxGtmLowerBoundOpposed*

For the proof consider the strategy matrix \( S(t) \) of \AName{} as described in \Cref{section:algorithm}.
Recapitulate that for one dimension the transition function for the outer robots simplifies as described in the beginning of \Cref{appendix:discreteOneDim}.
Therefore, the respective entries in \( S(t) \) collapse for the entries \( w_{2}(t) \) and \( w_{n}(t) \).
To achieve the addition of a constant for these vectors in each executed step, we change both \( w(t) \) and \( S(t) \).
\( w(t) \) is extended to $
w'(t) = \begin{bmatrix}
	1 & w_{2}(t) & \dots & w_{n}(t)
\end{bmatrix}^{T}.$
\( S(t) \) changes to the following matrix, independent of the current time $t$:
\begin{align*}
	A_{\theMatrixIndex} & = \begin{bmatrix}
		1      & 0      & 0      & 0      & 0      & 0      & \dots  & 0      & 0      \\
		\half  & 0      & \half  & 0      & 0      & 0      & \dots  & 0      & 0      \\
		0      & \half  & 0      & \half  & 0      & 0      & \dots  & 0      & 0      \\
		\vdots & \vdots & \vdots & \vdots & \vdots & \vdots & \vdots & \vdots & \vdots \\
		\half  & 0      & 0      & 0      & 0      & 0      & \dots  & \half  & 0
	\end{bmatrix}.
\end{align*}
Similar to before the behavior of \AName{} in one step is described by \( w'(t+1) = A_{\theMatrixIndex} \cdot w'(t) \).
Furthermore, \( A_{\theMatrixIndex} \) is stochastic and can also be interpreted as a Markov Chain with \( n \) states, where the first state is absorbing.
Similar to \cite{conf/spaa/KlingH11}, the following eigenvalues can be derived:
\begin{restatable}{lemma}{lemmaMaxGtmEigenvaluesLowerBound}
	\label{lemma:1-dimension:lower-bound-matrix-properties}
	\( A_{\theMatrixIndex} \) has the eigenvalues
	$\lambda_{j+1} = \cos \left(\frac{j \pi}{n}\right)$  for $j = 0, \dots, n - 1$ and the unique stationary distribution
	$\pi(A_{1}) = \begin{bmatrix}
		1 & 0 & \dots & 0 \end{bmatrix}$.

\end{restatable}

\begin{proof}[Proof of \Cref{theorem:1-dimension:lower-bound-opposed-configs}]
	Observe that by \Cref{lemma:1-dimension:lower-bound-matrix-properties}, the eigenvalues of \( A_{\theMatrixIndex} \) are equivalent to the ones used in \cite{conf/spaa/KlingH11}, since we have \( k = 1 \) and consider configurations given by \( n \) vectors in contrast to \( n + 1 \) vectors assumed in \cite{conf/spaa/KlingH11}.
	Due to \cite[Theorem 5]{conf/spaa/KlingH11}, for the spectral gap of \( A_{\theMatrixIndex} \) it holds
	$\lambda_{2} \in \Theta\left(1-\frac{1}{n^2}\right).$
	Combined with \Cref{theorem:1-dimension:lower-bound-mixing-time}, we can see that the mixing time for a factor of \( 2 \varepsilon \) we need a time of at least
	$t_{\text{mix}}(2\varepsilon) \geq \left(\frac{1}{1-\lambda_{2}} - 1 \right) \cdot \frac{1}{\log 4 \varepsilon}
	\in \Omega \left(n^{2} \log \left(\frac{1}{\varepsilon}\right)\right).$
	Assume we are at time step \( t<t_{\text{mix}}(2\varepsilon) \).
	Consider \( A_{\theMatrixIndex}^{t} \) and note that we can express the configuration at step \( t \) by \( w(t) = A_{\theMatrixIndex}^{t} \cdot w(0) \).
	Due to the mixing time and the stationary distribution of \( A_{\theMatrixIndex} \) introduced in \Cref{lemma:1-dimension:lower-bound-matrix-properties}, we know that there is an \( 1 \leq i \leq n \) such that \( A_{\theMatrixIndex}^{t}[i,1] < (1-2\varepsilon) \).
	Further, \( A_{\theMatrixIndex}^{t}[i,n] \leq 1 \).
	Define the initial configuration by
	$w_{1}(0) = 1, w_{i}(0) = -0.313 \text{ for } 2 \leq i \leq n$, and $w_{n}(0) = \varepsilon.$
	Notice that \( w(0) \) is a valid \opposedConfig{}.
	Also \( w_{i}(t+1) = \sum_{j=1}^{n} A_{\theMatrixIndex}^{t}[i,j] \cdot w_{j}(0) < (1-2\varepsilon) + \varepsilon = (1-\varepsilon) \) and thus, no \( \varepsilon \)-approx\-imation of the optimal configuration has been reached.
\end{proof}

\stepcounter{MatrixIndex}
\theoremMaxGtmMarchingBounds*

For the proof of \Cref{theorem:1-dimension:upper-bound-marching-configs}, we analyze the distance a robot moves in each time step.
Define $z_i(t) := p_i(t+1) - p_i(t)$ to be the vector pointing from $p_i(t)$ to $p_i(t+1)$.
W.l.o.g., we assume that the outer robots both move in positive direction.
More formally:
$z_1(t) =  \half p_1(t) + \half p_2(t) + \half - p_1(t),
z_i(t)    = \half p_{i-1}(t) + \half p_{i+1}(t) - p_i(t)  \textnormal{for } 1 < i < n$ and
$z_{n}(t) =    \half p_{n-1}(t) + \half p_ {n}(t) + \half - p_{n}(t) $.
Note that for an arbitrary $t \in \mathbb{N}_0$, $\sum_{i=1}^{n}z_i(t) = 1$.
The corresponding equations for the next time step can be computed as follows:
$z_1(t+1)    =   \half z_1(t)  + \half z_2(t),
z_i(t+1)    = \half z_{i-1}(t) + \half z_{i+1}(t)  \textnormal{ for } 1 < i < n$ and
$z_{n}(t+1) =    \half z_{n-1}(t) + \half z_{n}(t)$.
Now, let $z(t)$ be a column vector of length $n$ whose $i$-th entry contains $z_i(t)$.
For the next time step, $z(t+1)$ can be computed as the product of the transition matrix $A_{\theMatrixIndex}$ and $z(t)$:

\begin{align*}
	z(t+1) & = \begin{bmatrix}
		\half  & \half  & 0      & 0      & 0      & 0      & \dots  & 0      \\
		\half  & 0      & \half  & 0      & 0      & 0      & \dots  & 0      \\
		0      & \half  & 0      & \half  & 0      & 0      & \dots  & 0      \\
		0      & 0      & \half  & 0      & \half  & 0      & \dots  & 0      \\
		\vdots & \vdots & \vdots & \vdots & \vdots & \vdots & \vdots & \vdots \\
		0      & 0      & 0      & 0      & 0      & 0      & \half  & \half
	\end{bmatrix} \cdot z(t) = A_{\theMatrixIndex}^t \cdot z(1)
\end{align*}

Note that $A_{\theMatrixIndex}$ can also be interpreted as the transition matrix of an aperiodic and irreducible Markov chain.
The mixing time of this matrix has already been analyzed in \cite{conf/spaa/KlingH11}.

\begin{lemma}[\cite{conf/spaa/KlingH11}]
	For $A_{\theMatrixIndex}$,
	$t_{mix}(\varepsilon) \in \Omega \left( n^2 \log \frac{1}{\varepsilon} \right)	\textnormal{ and } 	t_{mix}(\varepsilon) \in \mathcal{O} \left( n^2 \log \frac{n}{\varepsilon} \right).$

\end{lemma}

Based on the results in \cite{conf/spaa/KlingH11}, \Cref{theorem:1-dimension:upper-bound-marching-configs} follows.

\subsection{Two-Dimensional Analysis}

\begin{restatable}{lemma}{lemmaMaxGtzmPotentialDifferenceTwoDim}\label{lemma:2-dimensions:potential-difference}

	$	\phi_{2}(t) - \phi_{2}(t+1)  \geq \quarter \sum_{i=1}^{n} \|z_{i-1}(t) - z_{i+1}(t)\|^2$

\end{restatable}

\begin{proof}

	\begin{align*}
		\phi_{2}(t+1) = \sum_{i=1}^{n} \|z_i(t+1) \|^2 & = \sum_{i=1}^{n} \| \frac{p_{i-1}(t+1) + p_{i+1}(t+1)}{2} - p_i(t+1) \|^2                                     \\
		& = \quarter \sum_{i=2}^{n-1} \| z_{i-1}(t) + z_{i+1}(t)\|^2                + \|z_1(t+1)\|^2 + \| z_n(t+1) \|^2
	\end{align*}

	With  $ \|z_1(t+1)\|^2 \leq \quarter \|z_1(t) + z_2(t) \|^2$ and
	$\|z_n(t~+~1)\|^2 \leq \quarter \|z_{n-1}(t) + z_n(t) \|^2$ it follows that
	$\phi_{2}(t+1)  \leq   \quarter \sum_{i=1}^{n} \| z_{i-1}(t) + z_{i+1}(t)\|^2$
	Now define $\Delta \phi_{2}(t)$ to be $\phi_{2}(t) - \phi_{2}(t+1)$
	With help of the parallelogram law, we derive a lower bound on $\Delta \phi_{2}(t)$:
	\begin{align*}
		\Delta \phi_{2}(t) & =  \sum_{i=1}^{n} \|z_i(t)\|^2 - \quarter  \sum_{i=1}^{n} \|z_{i-1}(t) + z_{i+1}(t) \|^2                              \\
		& = \sum_{i=1}^{n} \|z_i(t)\|^2  - \sum_{i=1}^{n}\left(\|z_{i}(t) \|^2 - \quarter \|z_{i-1}(t) - z_{i+1}(t) \|^2\right) \\
		& = \quarter \sum_{i=1}^{n} \|z_{i-1}(t) - z_{i+1}(t)\|^2
	\end{align*}
\end{proof}

\theoremConvergence*

\begin{proof}
	$\phi_2(t)$ is a monotonically decreasing function of $t$ which is bounded from below by $0$ and the potential difference can be lower bounded by $\quarter \sum_{i=1}^{n} \|z_{i-1}(t) - z_{i+1}(t)\|^2$ (\Cref{lemma:2-dimensions:potential-difference}).
	Hence, the potential difference can only be $0$ in case all vectors $z_i(t)$ are equal.
	Either, $z_1(t) = z_2(t) = \dots = z_n(t) > 0$ or $z_1(t) = \dots = z_n(t)= 0$.
	In the first case, all robots move the same distance into the same direction (a \marchingChain{}).
	In the second case, no robot moves at all.
	This can only be the case if the chain is stretched to a \maxChain{}.
\end{proof}

\stepcounter{MatrixIndex}
\theoremMaxGtmOneMovingUpper*
W.l.o.g.\ assume $r_1$ does not move at all and thus $p_1(0) = p_1(t)$ for all $t$.
We again use a potential function $\phi_2(t)$ which sums up the squared distances of robots to their target points.
This time, $z_1(t) = 0$ for all $t$, and therefore, we exclude $z_1(t)$ from the summation and obtain
$\phi_2(t) = 4  \sum_{i=2}^{n} \| z_i(t) \|^2$.
The new equation for $z_2(t+1)$ simplifies as follows (all other equations remain unchanged):
$\norm{z_2(t+1)}^2 =  \| \half p_1(t+1) + \half p_3(t+1) - p_2(t+1) \|^2
= \quarter \| z_3(t) \|^2 \leq \half \|z_3(t) \|^2$.

Define $z'(t) = \left(\norm{z_2(t)}^2, \dots, \norm{z_n(t)}^2 \right)^T$.
Given two $n$-di\-men\-sional column vectors $v$ and $v'$ we define $v \leq v'$ if  $v_i \leq v_i'$ for all $1 \leq i \leq n$.
Then, we can upper bound $z'(t+1)$ as follows:

\begin{align*}
	z'(t+1) & \leq
	\begin{bmatrix}
		0      & \half  & 0      & 0      & 0      & 0      & \dots  & 0      \\
		\half  & 0      & \half  & 0      & 0      & 0      & \dots  & 0      \\
		0      & \half  & 0      & \half  & 0      & 0      & \dots  & 0      \\
		0      & 0      & \half  & 0      & \half  & 0      & \dots  & 0      \\
		\vdots & \vdots & \vdots & \vdots & \vdots & \vdots & \vdots & \vdots \\
		0      & 0      & 0      & 0      & 0      & 0      & \half  & \half
	\end{bmatrix} \cdot z'(t) = A_{\theMatrixIndex} \cdot z'(t) = A_\theMatrixIndex^{t+1} \cdot z'(0)
\end{align*}

Observe that $A_\theMatrixIndex$ is a substochastic matrix in which every row except for the first one is stochastic.
To analyze the convergence time via \Cref{theorem:preliminaries:powers-of-substochastic-matrices}, we have to determine the largest eigenvalue and the eigenvectors of $A_{\theMatrixIndex}$.

\begin{restatable}{lemma}{lemmaMaxGtmEigenvaluesOneMoving} \label{lemma:2-dimensions:one-stationary-eigenvalues}
	The eigenvalues of $A_{\theMatrixIndex}$ are	$\lambda_{j} = \cos \left(\frac{(2j-1) \cdot \pi }{2n-1}\right) $  for $j = 1, \dots, n-1$.
	The corresponding eigenvectors are given by $x_j[i] = \cos \left( \frac{(2j-1) \cdot (2i-1)}{2 \cdot (2n-1)} \right)$ for $i,j = 1, \dots, n-1$ where $x_j[i]$ denotes the $i$-th entry of eigenvector $j$.
\end{restatable}

\begin{proof}
	The matrix $A_{\theMatrixIndex}$ is a special tridiagonal matrix whose eigenvalues and eigenvectors have been analyzed in \cite{yueh2005eigenvalues}.
	The matrices in \cite{yueh2005eigenvalues} can be generalized as

	\begin{align*}
		T(a,b,c,\alpha,\beta) =	\begin{bmatrix}
			- \alpha +b & c      & 0      & 0      & 0      & 0      & \dots  & 0          \\
			a           & b      & c      & 0      & 0      & 0      & \dots  & 0          \\
			0           & a      & b      & c      & 0      & 0      & \dots  & 0          \\
			0           & 0      & a      & b      & c      & 0      & \dots  & 0          \\
			\vdots      & \vdots & \vdots & \vdots & \vdots & \vdots & \vdots & \vdots     \\
			0           & 0      & 0      & 0      & 0      & 0      & a      & -\beta + b
		\end{bmatrix}.
	\end{align*}
	In our case, $a = c = \half, b= \alpha = 0, \beta = - \sqrt{ac} = - \half$.
	For these values of $a, b, c, \alpha$ and $\beta$, \Cref{lemma:2-dimensions:one-stationary-eigenvalues} follows from Theorem $2$ in \cite{yueh2005eigenvalues}.
\end{proof}

\begin{proof}[Proof of \Cref{theorem:2-dimensions:one-stationary-upper-bound}]
	We apply the results of \Cref{theorem:preliminaries:powers-of-substochastic-matrices}.
	Due to \Cref{lemma:2-dimensions:one-stationary-eigenvalues} it holds $\alpha \leq 1$ and $\beta = \cos \left( \frac{\pi}{2n-1}\right)$.
	Using $\cos(x) \leq 1 - \frac{2x^2}{\pi^2}$ for $-\pi \leq x \leq \pi$, we can derive $\beta \leq (1 - \frac{1}{(2n-1)^2})$.
	For $t = (2n-1)^2$ we obtain $\beta^t \leq \frac{1}{e}$.
	Hence, for $t' \geq (2n-1)^2 \cdot \ln \left( \frac{n^2}{\varepsilon} \right)$ it holds for all $i,j$ $A_\theMatrixIndex^{t'}[i,j] \leq \frac{\varepsilon}{n}$.
	Since initially $\|z_i(0)\|^2 \leq 1$ for all $1 \leq i \leq n-1$, we obtain for all $1 \leq i \leq n-1$: $\|z_i(t')\| \leq \varepsilon$.
	After $t'$ rounds, no robot moves a distance larger than $\varepsilon$ anymore.
	This can only be the case if all robots are almost aligned on a line and all vectors $w_i(t')$ are approximately the same.

\end{proof}

\deltaU*

For the proof of \Cref{theorem:2-dimensions:general-lower-bound}, we need two auxiliary lemmata:
\begin{restatable}{lemma}{lemmaMaxGtmLowerBoundIncreasing} \label{lemma:2-dimension:lower-bound-monotonically-increasing-outer-vectors}
	During an execution of \textnormal{\AName{}} starting in a \deltaUConfig{} at time step $t_0$, it holds that $\norm{w_i(t)} \geq \norm{w_i(t_0)} $ for all $t$ and all $2 \leq i \leq n$.
\end{restatable}

\begin{proof}
	Observe that for all $w_i(t)$ it always holds $x_i(t) \geq 0$ and for $1 \leq i \leq \frac{n+1}{2}-1: y_i(t) \leq 0$ and for
	$\frac{n+1}{2}-1 \leq i \leq n: y_i(t) \geq 0$.
	We prove the fact for vectors $w_{n/2+1}$ to $w_n$ by induction over $t$.
	The induction for the first half of vectors is analogous.
	The induction base is clear.
	Consider time step $t+1$ and an vector $w_i(t+1) = \half w_{i-1}(t) + \half w_{i+1}(t)$.
	It holds $\norm{w_i(t+1)} = \half \norm{w_{i-1}(t) + w_{i+1}(t)}$.
	Since $x_{i-1}(t), x_{i+1}(t), y_{i-1}(t)$ and $y_{i+1}(t) \geq 0$ it holds:
	$\norm{w_i(t+1)} \geq \half \norm{w_{i-1}(t_0) + w_{i+1}(t_0)} = \norm{w_i(t_0)}$.

\end{proof}

\begin{restatable}{lemma}{lemmaMaxGtmGeneralLowerBound} \label{lemma:2-dimensional-lower-bound-x}
	When applying \textnormal{\maxgtm/} to a \deltaUConfig{}, it holds that $x_2(t) \geq~\half$ after $\Omega\left(n^2 \cdot \log \left(1/\delta\right)\right)$ rounds.
\end{restatable}

\begin{proof}
	Trivially, $x_3(t) \leq x_2(t)$ for all $t$.
	Thus, we can bound $x_2(t+1)$ as follows.
	\begin{align*}
		x_2(t+1)  \leq \frac{1}{2-\frac{4}{n}}x_2(t) + \half x_3(t)   = \half x_2(t) + \frac{1}{n-2} x_2(t) + \half x_3(t) & \leq \left(1+ \frac{1}{n-2}\right) x_2(t)
	\end{align*}

	Thus, $x_2(t)$ doubles at most every $\mathcal{O}(n)$ rounds.
	Since $x_2(t_0) = \frac{\delta}{n-1}$ it requires at least $\mathcal{O}(n^2)$ rounds until $x_2(t) \geq 2 \delta$.
	The lemma follows.
\end{proof}

\Cref{lemma:2-dimensional-lower-bound-x} together with the lower bound for one-dimensional configurations imply \Cref{theorem:2-dimensions:general-lower-bound}.

\section{Relation to Discrete Dynamical Systems} \label{section:discreteDynamicalSystems}

The aim of this section is to prove that the \marchingChain{} is unstable fixed point of \maxgtm/, interpreted as a discrete dynamical system.
To do so, we split the vector representation of a configuration into its $x$- and $y$-components.
More precisely, consider the vector $w_i(t) := (x_i(t),y_i(t))$.
Now define the state of a system to be $s(t) := \left(x_2(t), x_3(t), \dots, x_n(t), y_2(t), y_3(t), \dots ,y_n(t)\right).$ The dynamical system consists of $2 \left(n-1\right)$ variables each representing an entry of the vector representation.
Applying \maxgtm/ to the configuration can now be interpreted as a set of $2 \left(n-1\right)$ functions, one function for each variable.
$
x_2(t+1) = f_{x_2}(s(t)) = \frac{x_2(t)}{2 \cdot \sqrt{x_2(t)^2 + y_2(t)^2}} + \half x_3(t)$,
$x_i(t+1) = f_{x_i}(s(t)) = \half x_{i-	1}(t) + \half x_{i+1}(t)$  for  $2 < i < n$,
$x_n(t+1)  = f_{x_n}(s(t)) = \half x_{n-1}(t)+ \frac{x_n(t)}{2 \cdot \sqrt{x_n(t)^2 + y_n(t)^2}}$,
$y_2(t+1) = f_{y_2}(s(t)) = \frac{y_2(t)}{2 \cdot \sqrt{x_2(t)^2 + y_2(t)^2}} + \half y_3(t)$,
$y_i(t+1) = f_{y_i}(s(t)) = \half y_{i-	1}(t) + \half y_{i+1}(t)$ for $2 < i < n$ and
$y_n(t+1) = f_{y_n}(s(t)) = \half y_{n-1}(t)+ \frac{y_n(t)}{2 \cdot \sqrt{x_n(t)^2 + y_n(t)^2}}$.

Next, we compute the \emph{Jacobian matrix} $\mathcal{J}$ of the dynamical system.
For this dynamical system, $\mathcal{J}$ is a $2 \cdot \left(n-1\right) \times 2 \cdot \left(n-1\right)$-matrix.
Each row corresponds to one of the $2 \cdot \left(n-1\right)$ transition functions.
An entry $\mathcal{J}_{i,j}$ represents $\frac{\partial f_{i}}{\partial j}$, the derivative of $f_i$ with respect to variable~$j$.
Each function depends on at most $3$ variables and thus each row contains at most $3$ non-zero elements.

For better readability, we omit the time parameter~$t$.
Plugging in the partial derivations leads to the following matrix:

\setcounter{MaxMatrixCols}{20}
\begin{align}
	\mathcal{J} & =
	\begin{bmatrix}
		\frac{y_2^2}{2 \left(x_2^2 + y_2^2\right)^{3/2}}          & \half  & 0      & \dots  & 0      & 0                                                         & \frac{-x_2 \cdot y_2}{2 \left(x_2^2 + y_2^2\right)^{3/2}} & 0      & \dots  & 0      & 0                                                         \\
		\half                                                     & 0      & \half  & \dots  & 0      & 0                                                         & 0                                                         & 0      & \dots  & 0      & 0                                                         \\
		\vdots                                                    & \vdots & \vdots & \vdots & \vdots & \vdots                                                    & \vdots                                                    & \vdots & \vdots & \vdots & \vdots                                                    \\
		0                                                         & 0      & 0      & \dots  & \half  & \frac{y_n^2}{2 \left(x_n^2 + y_n^2\right)^{3/2}}          & 0                                                         & 0      & \dots  & 0      & \frac{-x_n \cdot y_n}{2 \left(x_n^2 + y_n^2\right)^{3/2}} \\
		\frac{-x_2 \cdot y_2}{2 \left(x_2^2 + y_2^2\right)^{3/2}} & 0      & 0      & \dots  & 0      & 0                                                         & \frac{x_2^2}{2 \left(x_2^2 + y_2^2\right)^{3/2}}          & \half  & \dots  & 0      & 0                                                         \\
		0                                                         & 0      & 0      & \dots  & 0      & 0                                                         & \half                                                     & 0      & \dots  & 0      & 0                                                         \\
		\vdots                                                    & \vdots & \vdots & \vdots & \vdots & \vdots                                                    & \vdots                                                    & \vdots & \vdots & \vdots & \vdots                                                    \\
		0                                                         & 0      & \dots  & 0      & 0      & 0                                                         & 0                                                         & 0      & \dots  & 0      & \half                                                     \\
		0                                                         & 0      & \dots  & 0      & 0      & \frac{-x_n \cdot y_n}{2 \left(x_n^2 + y_n^2\right)^{3/2}} & 0                                                         & 0      & \dots  & \half  & \frac{x_n^2}{2 \left(x_n^2 + y_n^2\right)^{3/2}}          \\
	\end{bmatrix} \label{equation:jacobi}
\end{align}

To prove that the \marchingChain{} is an unstable fixed point, we have to evaluate $\mathcal{J}$ at that fixed point.
Recall the \marchingChain{} is one-dimensional.
Thus assume that $x_i(t) = 0$ for $2 \leq i \leq n$.
The variables $y_i(t)$ are defined according to the marching chain: $y_2(t) = 1- \frac{2}{n}$, $y_3(t) = 1-\frac{4}{n}, \dots, y_n(t) = - \left(1-\frac{2}{n}\right)$.
Plugging these values into the Jacobian matrix (\Cref{equation:jacobi}) yields the following matrix:

\begin{align}
	\mathcal{J}_{w_M} & =
	\begin{bmatrix}
		\frac{n}{2 \left(n-2 \right)} & \half  & 0      & \dots  & 0      & 0                             & 0      & 0      & \dots  & 0      & 0      \\
		\half                         & 0      & \half  & \dots  & 0      & 0                             & 0      & 0      & \dots  & 0      & 0      \\
		\vdots                        & \vdots & \vdots & \vdots & \vdots & \vdots                        & \vdots & \vdots & \vdots & \vdots & \vdots \\
		0                             & 0      & 0      & \dots  & \half  & \frac{n}{2 \left(n-2 \right)} & 0      & 0      & \dots  & 0      & 0      \\
		0                             & 0      & 0      & \dots  & 0      & 0                             & 0      & \half  & \dots  & 0      & 0      \\
		0                             & 0      & 0      & \dots  & 0      & 0                             & \half  & 0      & \dots  & 0      & 0      \\
		\vdots                        & \vdots & \vdots & \vdots & \vdots & \vdots                        & \vdots & \vdots & \vdots & \vdots & \vdots \\
		0                             & 0      & \dots  & 0      & 0      & 0                             & 0      & 0      & \dots  & 0      & \half  \\
		0                             & 0      & \dots  & 0      & 0      & 0                             & 0      & 0      & \dots  & \half  & 0      \\
	\end{bmatrix} \label{equation:jacobiMarching}
\end{align}

The stability of the \marchingChain{} can now be analyzed bei the eigendecomposition of $\mathcal{J}_{w_M}$.
More precisely, in case at least one eigenvalue of $J_{w_M}$ has a magnitude larger than $1$ it follows that the marching chain is an unstable fixed point, as stated by the following theorem.

\begin{theorem}[\cite{galor2007discrete}]
	Let $f : D \rightarrow \mathcal{R}^{m},D \subseteq \mathcal{R}^{m}$ be a continuous, differentiable map with regard to
	all system state variables defined on an open subset around a fixed point $x^{*}$ and let $A$ be the
	Jacobian matrix of the system about $x^{*}$ .
	Then:
	\begin{enumerate}
		\item If the maximum modulus of the eigenvalues of $A$ is less than one then $x^{*}$ is asymptotically stable.
		\item If the maximum modulus of the eigenvalues of $A$ is greater than one then $x^{*}$ is unstable.
		\item If the maximum modulus of the eigenvalues of $A$ is equal to one then no conclusion is drawn.
	\end{enumerate}

\end{theorem}
The following lemma helps us to prove a lower bound on the largest eigenvalue of $\mathcal{J}_{w_M}$.

\begin{lemma}[\cite{WALKER2008519}] \label{lemma:lowerBoundEigenvalue}
	Define $u^{T} = \left(1, \dots, 1\right)$. The largest eigenvalue $\lambda_1(A)$ of a symmetric $n \times n$ matrix $A$ can be lower bounded by $\lambda_1(A) \geq \frac{u^{T} \cdot A \cdot u}{u^{T} \cdot u}$

\end{lemma}

\begin{theorem}
	The \marchingChain{} is an unstable fixed point.
\end{theorem}

\begin{proof}
	As $\mathcal{J}_{w_M}$ is a $2 \cdot \left(n-1\right) \times 2 \cdot \left(n-1\right)$ matrix, it follows $u^{T} \cdot u = 2 \cdot \left(n-1\right)$.
	Next observe that every column of $\mathcal{J}_{w_m}$, except for the first and the $n$-th, are stochastic and thus sum up to $1$.
	The first and the $n$-th column sum up to $\half + \frac{n}{2 \cdot \left(n-2\right)} = 1 + \frac{1}{n-2}$.
	Thus, we can compute
	$u^{T} \cdot A  = \left(1 + \frac{1}{n-2},1, 1, \dots, 1, 1 + \frac{1}{n-2},1, 1, \dots, 1\right)$ and $u^{T} \cdot A \cdot u= 2 \cdot (n-1)  + \frac{2}{n-2}$.

	Finally, we can bound $\lambda_1(\mathcal{J}_{w_M})$ via \Cref{lemma:lowerBoundEigenvalue} and prove the theorem:
	$\lambda_1(\mathcal{J}_{w_M}) \geq \frac{2\cdot \left(n-1\right) + \frac{2}{n-2}}{2 \cdot \left(n-1\right)} = 1 + \frac{1}{n-1} + \frac{1}{n-2} > 1.$

\end{proof}

\section{Analysis of \parameterizedMaxMoveOnBisector{1}{1-\texorpdfstring{\tau}{t}}} \label{section:contTwoDim}

We restate a lemma that can be found in \cite{journals/topc/DegenerKKH15} that expresses how the distance between robots changes based on their velocity vectors.
For this purpose, we define the angles $\beta_{i, j}(t) \coloneqq \angle\bigl( v_i(t), p_j(t) - p_i(t) \bigr)$.
In other words, $\beta_{i, j}(t)$ denotes the signed angle between the velocity vector $r_i$ and the line segment connecting $r_i$ and $r_j$ at time $t$.

\begin{lemma} [Lemma 3.1 in \cite{journals/topc/DegenerKKH15}] \label{lemma:continuousDistances}
	Consider two robots $r_i$ and $r_j$ and let $\Delta_{i,j}(t) : \mathbb{R}_{\geq 0} \rightarrow \mathbb{R}_{\geq 0}$ represent their distance at time $t$.
	The distance between $r_i$ and $r_j$ changes with speed
	$\deriv{ \Delta_{i,j}}{t}	 = - (\norm{v_i(t)} \cdot \cos (\beta_{i, j}(t))  + \norm{v_j(t)} \cdot \cos (\beta_{j,i}(t))). $

\end{lemma}
\begin{restatable}{lemma}{lemmaMaxMobPotentialDecreasing} \label{lemma:potentialDecreasing}
	$\ell(t)$ is mon. increasing and $r(t)$ is mon. decreasing until $\ell(t) = r(t)$.
\end{restatable}

\begin{proof}
	We prove that $\ell(t)$ is monotonically increasing.
	With the same arguments it follows that $r(t)$ is monotonically decreasing.
	First of all, every robot that is located on the straight line between its neighbors or on the same position with at least one of its neighbors will never compute a target point that does not lie between its neighbors.
	Additionally, it will also not compute a target point to the left of its left neighbor or to the right of its right neighbor.
	Additionally, all the robots follow the movements of their neighbors and thus all robots between $r_0$ and $r_{\ell(t)}$ stay on the straight line connecting $r_0$ and $r_{\ell(t)}$.
	Hence, none of these robots causes a decrease of $\ell(t)$.
\end{proof}

\begin{restatable}{lemma}{lemmaMaxMobInnerLengthDecreasing} \label{lemma:vectorsInBetween}
	$\innerLength{}$ is monotonically decreasing.
\end{restatable}

\begin{proof}
	An inner robot $r_j$ can execute three different movements.
	Either it does not move at all, it moves with speed $1$ along the bisector between its neighbors, or it follows the movements of the midpoint of its neighbors.
	Non moving robots do not increase the length of any vector.
	Robots $r_j$ that move along the bisector decrease both $\norm{w_{j}(t)}$ and $\norm{w_{j+1}(t)}$ with speed $\cos \left(\frac{\alpha_{j}(t)}{2}\right) > 0$.
	This is a conclusion from \Cref{lemma:continuousDistances}  since $\beta_{j, j-1}(t) = \beta_{j, j+1}(t) = \frac{\alpha_{j}(t)}{2}$ and $r_i$ moves with speed $1$.
	Robots that follow the movements of their neighbors also cannot increase the length of any vector because they neither follow the movement of $r_1$ nor the movement of $r_{n}$.
	Thus, they follow robots that decrease the lengths of neighboring vectors (or do not move at all).
	Since all possible movements cause a decrease of $\innerLength{}$, $\innerLength{}$ is monotonically decreasing.
\end{proof}

\begin{restatable}{lemma}{lemmaMaxMobOuterAnglesSmall} \label{lemma:smallAnglesInnerLength}
	Consider a configuration fulfilling $\ell(t) \neq r(t)$ with an outer angle $\alpha_{i}(t) < \referenceAngle{}$, for $i \in \left\{\ell(t),r(t)\right\}$.
	Then, $\innerLength{}$ decreases with speed at least $\cos\bigl(\frac{\alpha_{i}(t)}{2}\bigr)$.
\end{restatable}

\begin{proof}
	The robot $r_i$ moves with speed $1$ along the bisector of vectors pointing to its neighbors.
	This movement decreases $\Delta_{i,i^{+}}(t)$ with speed $\cos \left(\frac{\alpha_{i}(t)}{2}\right)$.
	This can be derived from \Cref{lemma:continuousDistances} by observing $\beta_{i,i^{+}}(t) = \frac{\alpha_{i}(t)}{2}$.
	At the same time, $r_{i^{+}(t)}$ is defined such that $\alpha_{i^{+}}(t) < \pi$ and thus $r_{i^{+}(t)}$ either does not move at all or it moves also along the bisector between its neighbors.
	Hence, the movement of $r_{i^{+}(t)}$ can also not increase $\Delta_{i,i^{+}}(t)$.
	All robots in between stay on the straight line between $r_{i(t)}$ and $r_{i^{+}(t)}$.
	Lastly, observe $\Delta_{i,i^{+}}(t)$ is part of $\innerLength{}$ which proves the lemma.
\end{proof}

While $\innerLength{}$ can only decrease (\Cref{lemma:vectorsInBetween}), $\leftOuter{}$ and $\rightOuter{}$ can either increase or decrease, depending on the current size of outer angles.

\begin{restatable}{lemma}{lemmaMaxMobOuterLengthSmallAngle} \label{lemma:outerAnglesSmallVectors}
	Consider a configuration with $\ell(t) \neq r(t)$ and an outer angle fulfilling $0 \leq \alpha_{i}(t) < \referenceAngle{}$, for $i \in \left\{\ell(t), r(t)\right\}$.
	Then, $- \tau < \deriv{O_i}{t} \leq 0$.
\end{restatable}

\begin{proof}
	Assume $i = \ell(t)$ and let us analyze $\deriv{\leftOuterShort{}}{t}$.
	In this configuration, $\beta_{1, \ell}(t) = \pi$ and $\beta_{\ell, 1}(t) = \frac{\alpha_{\ell}(t)}{2}$.
	Since $\alpha_{\ell}(t) < \referenceAngle{}$ it follows $\cos \left(\beta_{\ell, 1}(t)\right) > 1-\tau$.
	Thus, the outer robot moves at full speed, i.e.\ $\norm{v_1(t)} = 1-\tau$.
	\Cref{lemma:continuousDistances} yields $\deriv{\leftOuterShort{}}{t} = - \left(-\left(1-\tau\right) + \cos \left(\frac{\alpha_{\ell}(t)}{2}\right)\right) = 1 - \tau -\cos \left(\frac{\alpha_{\ell}(t)}{2}\right)$.
	Lastly, we observe $1-\tau < \cos \left(\frac{\alpha_{\ell}(t)}{2}\right) \leq 1$ and conclude $- \tau $ $< \deriv{O_i}{t}$ $\leq 0$.
\end{proof}

In case the outer angle $\alpha_{i}(t)$ has a size of at least $\referenceAngle{}$, the robot $r_i$ does not move at all in case $O_i(t) < \gamma_i(t)$.
As a consequence, the length of $O_i(t)$ increases with constant speed.

\begin{restatable}{lemma}{lemmaMaxMobOuterLengthDecreasing}  \label{lemma:outerAngleMediumVectors}
	Consider a configuration with an outer angle fulfilling $\alpha_{i}(t) \geq \referenceAngle{}$ and $O_i(t) < \gamma_i(t)$.
	Then, $ \deriv{O_i}{t} = 1 - \tau$.
\end{restatable}

\begin{proof}
	W.l.o.g.\ assume that $i = \ell(t)$.
	We prove that $r_{\ell(t)}$ does not move at all in this case.
	In case $\|w_{\ell}(t)\| < 1$ the robot $r_{\ell(t)}$ does not move at all due to the definition of the strategy.
	Consider the case $\norm{w_{\ell}(t)} = 1$.
	Since $O_i(t) < \gamma_i(t)$, there must be at least one vector with index less than $\ell(t)$ that has a length smaller than $1$.
	Let $j$ be the highest index of a such a vector so that $\norm{w_i(t)} = 1$ for all $j < i < \ell(t)$.
	The robot $r_j$ moves with speed $1$ into the direction of $r_{j+1}$ as  $\norm{w_j(t)} < \norm{w_{j+1}(t)}$ and therefore the midpoint between $r_{j-1}$ and $r_{j+1}$ must lie closer to $r_{j+1}$.
	Note that $\beta_{j, \ell}(t) = 0$ and thus, by \Cref{lemma:continuousDistances}, $\deriv{\Delta_{j,\ell}}{t} = -1$.
	Since all robots with an index between $j$ and $\ell(t)$ follow the movement of the midpoint between their neighbors, it follows that all vectors $w_i(t)$ for an index $j < i \leq \ell(t)$ decrease in length.
	Hence, after an infinitesimal time interval it also holds $\norm{w_{\ell}(t)} < 1$ and thus $r_\ell(t)$ does not move at all.

	At the same time, $r_1$ is able to move with speed $1-\tau$ away from $r_{\ell(t)}$ and increases the distance between $r_1$ and $r_{\ell(t)}$ with speed $1-\tau$.
	To see this, we again have to consider two cases.
	In case $\|w_{2}(t)\| < 1$, $r_1$ moves with speed $1-\tau$ due to the definition of the strategy.
	In the other case, namely $\norm{w_2(t)} = 1$ consider the smallest index $k < \ell(t)$ such that $\|w_i(t)\| = 1$ for all $1 < i < k$ and $\norm{w_{k}(t)} < 1$.
	The robot $r_{k-1}$ moves with speed $1$ towards the midpoint between $r_{k-2}$ and $r_{k}$.
	Due to the definition of $k$, this midpoint lies closer to $r_{k-2}$ than to $r_k$.
	This implies $\beta_{k, 1}(t) = 0$ and by \Cref{lemma:continuousDistances} it follows that the movement of $r_{k}$ decreases $\Delta_{k,1}(t)$ with speed $1$.
	Since all robots with an index between $1$ and $k$ follow the movement of the midpoint between their neighbors, it follows that all vectors $w_i(t)$ for an index $2 \leq i \leq k$ decrease in length.
	This implies that $\norm{w_2(t)}$ decreases after an infinitesimal time interval such that $r_1$ is able to move at its maximum speed $1-\tau$.

	Combining the movements of $r_1$ and $r_{\ell(t)}$ it follows that $r_{\ell(t)}$ does not move at all, while $r_{1}$ moves with speed $1-\tau$ with an angle $\beta_{1, \ell}(t) = \pi$ such that
	$\deriv{\leftOuterShort{}}{t} = 1-\tau$ by \Cref{lemma:continuousDistances}.
	The arguments for $\rightOuter{}$ are analogous.

\end{proof}

In configurations with an outer angle of a size at most $\psi$ it holds that $I(t)$ decreases with speed $1-\tau$.
Therefore, the total time such a configuration can exist ist upper bounded by $\frac{n-3}{1-\tau}$ since $I(t)$ is upper bounded by $n-3$.

\begin{restatable}{lemma}{lemmaMaxMobInnerLength} \label{lemma:innerLengh}
	In configurations having an outer angle $\alpha_{i}(t) < \referenceAngle{}$, for $i \in \left\{\ell(t).r(t)\right\}$, it holds $\innerLength{}$ decreases with speed at least $1- \tau$.
\end{restatable}

\begin{proof}
	W.l.o.g.\ assume $i = \ell(t)$, the arguments for $i = r(t)$ are analogous.
	In such a configuration, $\ell(t)$ moves with speed $1$ along the bisector of vectors pointing to its neighbors.
	At the same time, $r_{\ell^{+}}(t)$ either does not move or it moves with speed $1$ along its bisector (provided $\alpha_{\ell^{+}}(t) < \referenceAngle{}$).
	In case it does not move, $\Delta_{\ell,\ell^{+}}(t)$ decreases with speed $\cos \left(\frac{\alpha_{\ell}(t)}{2}\right) \geq 1- \tau$, as $\beta_{\ell, \ell^{+}}(t) = \frac{\alpha_{\ell}(t)}{2}$.
	In case both move, $\Delta_{\ell,\ell^{+}}(t)$ decreases with speed $\cos \left(\frac{\alpha_{\ell}(t)}{2}\right) + \cos \left(\frac{\alpha_{\ell^{+}}(t)}{2}\right) \geq 2 \cdot \left(1- \tau\right)$.
	This can be derived from \Cref{lemma:continuousDistances} since $\beta_{\ell, \ell^{+}}(t) = \frac{\alpha_{\ell}(t)}{2}$ and $\beta_{\ell^{+}, \ell}(t) = \frac{\alpha_{\ell^{+}}(t)}{2}$ and $\norm{v_{\ell}(t)} = \norm{v_{\ell^{+}}(t)}$ $ = 1$.
	Since $\Delta_{\ell,\ell^{+}}(t)$ is part of $\innerLength{}$  wen can conclude that $\innerLength{}$ decreases with speed at least $1-\tau$.
\end{proof}

\begin{corollary} \label{lemma:turnaroundTotal}
	The total time an outer angle of size less than \referenceAngle{} exists, while $\ell(t) \neq r(t)$, is bounded by $
	\frac{n-3}{1-\tau}$.
\end{corollary}

Next, we analyze the behavior of outer angles that have a size of at least \referenceAngle{}.
A robot corresponding to an outer angle $\alpha_{i}(t) \geq \referenceAngle{}$ only moves in case $O_i(t) = \gamma_i(t)$.
The following lemmata assume that $O_i(t) = \gamma_i(t)$.

\begin{restatable}{lemma}{lemmaMaxMobOuterAngleMedium} \label{lemma:outerAngleMedium}
	Assume that $i \in \left\{\ell(t), r(t)\right\}$.
	In configurations with an outer angle of size $\alpha_{i}(t): \referenceAngle{} \leq \alpha_{i}(t) \leq \frac{3}{4} \pi$ while $O_{i}(t) = \gamma_i(t)$ and $\ell(t) \neq r(t): \deriv{\innerLengthShort{}}{t} \leq - \frac{\sqrt{2-\sqrt{2}}}{2}$ .
\end{restatable}

\begin{proof}
	W.l.o.g. assume that $i = \ell(t)$.
	Since $\leftOuter{} = \gamma_\ell(t)$, it holds $\norm{w_{\ell}(t)} = 1$ and thus $r_{\ell(t)}$ moves with speed $1$ along the bisector formed by vectors pointing to its neighbors.
	By noticing $\beta_{\ell,1}(t) = \frac{\alpha_{\ell}(t)}{2}$ and applying \Cref{lemma:continuousDistances}, we conclude that the movement of $r_\ell(t)$ decreases $\Delta_{\ell,\ell^{+}}(t)$ with speed at least $\cos \left(\frac{3}{8} \pi\right) = \frac{\sqrt{2-\sqrt{2}}}{2}$ and at most $\cos \left(\frac{\referenceAngle{}}{2}\right) = 1-\tau$.
	As $r_1$ is able to move with speed at most $1- \tau$, $r_1$ moves fast enough such that the distance between $r_1$ and $r_{\ell(t)}$ does not decrease and thus $\leftOuter{}$ and especially $\norm{w_{\ell}(t)}$remains constant.
	Hence, $r_{\ell(t)}$ continues moving along its bisector while decreasing $\innerLength{}$ with speed at least $\frac{\sqrt{2-\sqrt{2}}}{2}$.

\end{proof}

\begin{corollary} \label{corollary:mediumSizeTime}
	The total time an outer angle of size $\referenceAngle{} \leq \alpha_{i}(t) \leq \frac{3}{4}\pi$ exists, while $O_i(t) = \gamma_i(t)$ and $\ell(t) \neq r(t)$ is bounded by $\frac{2 \cdot \left(n-3\right)}{\sqrt{2-\sqrt{2}}}$.
\end{corollary}

It remains to analyze outer angles that have a size of at least $\frac{3}{4}\pi$.
It turns out that an outer angle of size at least $\frac{3}{4}\pi$ only increases.

\begin{restatable}{lemma}{lemmaMaxMobOuterAnglesIncreasing} \label{lemma:outerAngleIncreasing}
	Assume that $i \in \left\{\ell(t), r(t)\right\}$.
	In configurations fulfilling $\ell(t) \neq r(t)$, $\alpha_{i}(t) \geq \frac{3}{4} \pi$ and $O_i(t) = \gamma_i(t)$ for an outer angle $\alpha_{i}(t)$ it holds that $\alpha_{i}(t)$ is monotonically increasing.
\end{restatable}

\begin{proof}
	We give the proof for $\alpha_{\ell}(t)$.
	For this, we rewrite $\alpha_{\ell}(t) = \pi -c$ for $0 \leq c < \frac{1}{4}\pi$.
	Let $f_{\ell}(t) = \cos \left(\alpha_{\ell}(t)\right)$.
	We compute the derivation $\deriv{f_{\ell}}{t}$ and prove $\deriv{f_{\ell}}{t} < 0$.
	As $\cos \left(x\right)$ is monotonically decreasing in the interval $[\frac{3}{4}\pi, \pi)$, this proves that $\alpha_{i}(t)$ is monotonically increasing.
	Let $\beta^{-}(t)$ be the angle enclosed by the line segments connecting $r_1$ and $r_{\ell^{+}(t)}$ and $r_1$ and $r_{\ell(t)}$. See \Cref{lemma:outerAngleIncreasing} for a visualization.
	Similarly, let $\beta^{+}(t)$ denote the angle enclosed by the line segments connecting $r_1$ and $r_{\ell^{+}(t)}$ and $r_{\ell^{+}(t)}$ and $r_{\ell(t)}$.
	\begin{figure}[htb]
		\centering
		\includegraphics[width=0.52\textwidth, clip=true]{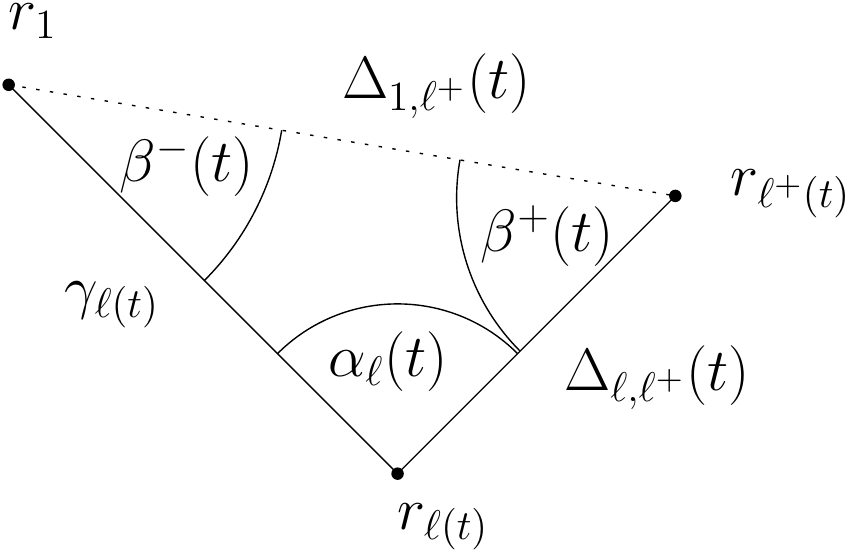}
		\caption{Visualization of the definitions needed for the proof of \Cref{lemma:outerAngleIncreasing}.}
		\label{figure:triangleHeight}
	\end{figure}

	We start with giving a formula for $f_{\ell}(t)$ and compute its derivation.
	Note that $\leftOuter{}$ stays constant as $\beta_{\ell,1}(t) = \frac{\alpha_{\ell}(t)}{2}$ and thus the movement of $r_{\ell(t)}$ decreases $\leftOuter{}$ with speed at most $\cos \left(\frac{3}{8}\right)$.
	Since $\cos \left(\frac{3}{8}\right) < 1- \tau$, $r_1$ moves fast enough such that $\leftOuter{} = \gamma_\ell(t)$ remains constant.

	Now consider the triangle formed by $r_1$, $r_{\ell(t)}$ and $r_{\ell^{+}(t)}$.
	By our assumption, $\Delta_{1,\ell}(t) = \gamma_\ell(t)$.
	Via the law of cosines, we obtain: $\Delta_{1,\ell^{+}}(t)^2 = \gamma_{\ell}(t)^2 + \Delta_{\ell,\ell^{+}}(t)^2 - 2 \cdot \Delta_{1,\ell^{+}}(t) \cdot \Delta_{\ell,\ell^{+}}(t) \cdot \cos \left(\alpha_{\ell}(t)\right)$.
	By substituting $\cos \left(\alpha_{\ell}(t)\right)$ by $f_{\ell}(t)$ and rearranging the terms, we get the following formula for $f_{\ell}(t)$:
	$f_\ell(t) =  \frac{\gamma_\ell(t)^2 + \Delta_{\ell,\ell^{+}}(t)^2  - \Delta_{1,\ell^{+}}(t)^2}{2 \cdot \gamma_\ell(t) \cdot \Delta_{\ell,\ell^{+}}(t)}$.

	Now, we compute $\deriv{f_{\ell}}{t}$.
	Remember that $\deriv{\gamma_{\ell}}{t} = \gamma_{\ell}(t)$ as stated above.
	$$\deriv{f_{\ell}}{t} = \frac{\deriv{\Delta_{\ell,\ell^{+}}}{t} \cdot \left(-\gamma_\ell(t)^2 +  \Delta_{\ell,\ell^{+}}(t)^2 + \Delta_{1,\ell^{+}}(t)^2\right)}{2 \cdot \gamma_\ell(t) \cdot \Delta_{\ell,\ell^{+}}(t)^2}
	- \frac{2 \cdot \Delta_{\ell,\ell^{+}}(t) \cdot \Delta_{1,\ell^{+}}(t)\deriv{\Delta_{\ell,\ell^{+}}}{t}}{2 \cdot \gamma_\ell(t) \cdot \Delta_{\ell,\ell^{+}}(t)^2}$$.
	Applying the law of cosines again gives us $-\gamma_\ell(t)^2 + \Delta_{1,\ell^{+}}(t)^2 + \Delta_{\ell,\ell^{+}}(t)^2 = 2 \cdot \Delta_{\ell,\ell^{+}}(t) \cdot \Delta_{1,\ell^{+}}(t) \cdot \cos \left(\beta^{+}(t)\right)$.
	Replacing this in the original formula for $\deriv{f_{\ell}}{t}$ yields:
	$\deriv{f_{\ell}}{t} = \frac{\Delta_{1,\ell^{+}}(t)}{\gamma_\ell(t) \cdot \Delta_{\ell,\ell^{+}}(t)} \cdot \left(\deriv{\Delta_{\ell,\ell^{+}}}{t}  \cdot \cos \left(\beta^{+}(t) \right)-\deriv{\Delta_{1,\ell^{+}}}{t} \right)$.
	Now, we have to consider two cases.
	Either, $r_{\ell^{+}(t)}$ is moving and thus $\alpha_{\ell^{+}}(t) \leq \referenceAngle{}$ or $r_{\ell^{+}(t)}$ does not move.
	In both cases it holds that $ \Delta_{\ell,\ell^{+}}(t)  \leq 0$ as $\ell(t) \neq r(t)$ (\Cref{lemma:vectorsInBetween}).
	We start by analyzing the case that $r_{\ell^{+}(t)}$ does not move.
	Observe $\beta_{1, \ell^{+}}(t) = \pi - \beta^{-}(t)$.
	By \Cref{lemma:continuousDistances}, we obtain $\deriv{\Delta_{1,\ell^{+}}}{t} = - \norm{v_1(t)} \cdot \cos \left(\pi-\beta^{-}(t)\right) > 0$ as $\beta^{-}(t)$ can be upper bounded by $\frac{1}{4} \pi$ and thus $\pi - \beta^{-}(t) > \frac{3}{4}\pi$.
	We can conclude,  $-\deriv{\Delta_{1,\ell^{+}}}{t} < 0$ and $\deriv{f_{\ell}}{t} < 0$ in this case.
	As both  $\Delta_{\ell,\ell^{+}}(t) \leq 0$ and $-\deriv{\Delta_{1,\ell^{+}}}{t} < 0$ it follows $\deriv{f_{\ell}(t)}{t} < 0$.
	It remains to analyze the case that $r_{\ell^{+}}(t)$ is moving.
	We compute $\deriv{\Delta_{1,\ell^{+}}}{t}$.
	Depending on the orientation of $\alpha_{\ell^{+}}(t)$, there are two possible variants of $\deriv{\Delta_{1,\ell^{+}}}{t}$.
	Either $\alpha_{\ell}(t)$ and $\alpha_{\ell^{+}}(t)$ have the same or different orientations.
	Consider the case that $\alpha_{\ell}(t)$ and $\alpha_{\ell^{+}}(t)$ have the same orientation.
	In this case
	$\deriv{\Delta_{1,\ell^{+}}}{t} = \cos \left(\frac{\alpha_{\ell}(t)}{2}\right) \cdot \cos \left(\beta^{-}(t)\right) - \cos \left(\frac{\alpha_{\ell^{+}}(t)}{2} - \beta^{+}(t)\right)$.
	In the other case, it holds	$\deriv{\Delta_{1,\ell^{+}}}{t} = \cos \left(\frac{\alpha_{\ell}(t)}{2}\right) \cdot \cos \left(\beta^{-}(t)\right) - \cos \left(\frac{\alpha_{\ell^{+}}(t)}{2} + \beta^{+}(t)\right)$.
	Since 	$\cos \left(\frac{\alpha_{\ell}(t)}{2}\right) \cdot \cos \left(\beta^{-}(t)\right) - \cos \left(\frac{\alpha_{\ell^{+}}(t)}{2} + \beta^{+}(t)\right) > \cos \left(\frac{\alpha_{\ell}(t)}{2}\right) \cdot \cos \left(\beta^{-}(t)\right) - \cos \left(\frac{\alpha_{\ell^{+}}(t)}{2} - \beta^{+}(t)\right)$ , we can analyze the first case which immediately implies the second case.
	Thus, assume that $\deriv{\Delta_{1,\ell^{+}}}{t} = \cos \left(\frac{\alpha_{\ell}(t)}{2}\right) \cdot \cos \left(\beta^{-}(t)\right) - \cos \left(\frac{\alpha_{\ell^{+}}(t)}{2} - \beta^{+}(t)\right).$
	Note that $\cos \left(a-b\right) = \sin\left(a\right) \cdot \sin \left(b\right) + \cos \left(a\right) + \cos \left(b\right)$ and thus  $\cos \left(\frac{\alpha_{\ell^{+}}(t)}{2} - \beta^{+}(t)\right) = \sin \left(\frac{\alpha_{\ell^{+}}(t)}{2}\right) \cdot \sin \left(\beta^{+}(t)\right)  + \cos \left(\frac{\alpha_{\ell^{+}}(t)}{2}\right) \cdot \cos \left(\beta^{+}(t)\right) $.
	Additionally, we obtain via \Cref{lemma:continuousDistances}, $\deriv{\Delta_{\ell,\ell^{+}}}{t} = - \left(\cos \left(\frac{\alpha_{\ell}(t)}{2}\right) + \cos \left(\frac{\alpha_{\ell^{+}}(t)}{2}\right)\right)$.
	For improved readability, $\sigma(t) := \frac{\Delta_{1,\ell^{+}}(t)}{\gamma_\ell(t) \cdot  \Delta_{\ell,\ell^{+}}(t)}$ and
	$\left(\cos \left(\beta^{-}(t)\right) + \cos \left(\beta^{+}(t)\right)\right) =: \mu(t)$.

	Plugging all these insights into $\deriv{f_{\ell}}{t}$ allows us the following estimation (see below).
	\Cref{equation:sineAcos} holds as $\sin \left(\cos^{-1}\left(x\right)\right) = \sqrt{1-x^2}$.
	Additionally, it holds that $\cos \left(\frac{\pi}{2} -x\right) = \sin \left(x\right)$ and thus $\cos \left(\frac{\pi}{2} - \frac{c}{2}\right) = \sin \left(\frac{c}{2}\right)$.
	For \Cref{equation:cosineVodoo}, note $\beta^{-}(t) = \pi - \alpha_{\ell}(t) - \beta^{+}(t)$ as the sum of internal angles of a triangle is equal to $\pi$.
	Hence, we can rewrite $\cos \left(\beta^{-}(t)\right) = \cos \left(\pi - \alpha_{\ell}(t) - \beta^{+}(t)\right)  = \cos \left(\pi - \pi + c - \beta^{+}(t)\right) = \cos \left(c-\beta^{+}(t)\right)$.
	Now observe that $\cos \left(\beta^{+}(t)\right) + \cos \left(c- \beta^{+}(t)\right) = 2 \cdot \cos \left(\frac{c}{2}\right) \cdot \cos \left(\frac{c}{2} - \beta^{+}(t)\right)$.
	As a last step we use the equality $2 \cdot \cos \left(\frac{x}{2}\right) \cdot \sin \left(\frac{x}{2}\right) = \sin \left(x\right)$.
	Plugging all together yields $\sin \left(\frac{c}{2}\right) \cdot \mu(t) = \sin \left(c\right) \cdot \cos \left(\frac{c}{2} - \beta^{+}(t)\right)$.
	Finally, we can conclude that $f_\ell(t)$ is monotonically decreasing for $\alpha_{\ell}(t) \in \left[\frac{3}{4}\pi, \pi\right]$ and, thus, that $\alpha_{\ell}(t)$ is monotonically increasing in the same interval.
	\begin{align}
		& \deriv{f_{\ell}}{t}  =  \sigma(t) \, \left(\sin \left(\frac{\alpha_{\ell^{+}}(t)}{2}\right) \cdot \sin \left(\beta^{+}(t)\right) - \cos \left(\frac{\alpha_{\ell}(t)}{2}\right) \cdot \mu(t) \right)     \\
                     	& \leq   \sigma(t) \, \left(\sin \left( \referenceAngle{}\right) \cdot \sin \left(\beta^{+}(t)\right) - \cos \left(\frac{\alpha_{\ell}(t)}{2}\right) \cdot \mu(t) \right)                                  \\
		& =  \sigma(t) \, \left( \sqrt{1-\left(1-\tau\right)^2} \cdot \sin \left(\beta^{+}(t)\right) - \cos \left(\frac{\alpha_{\ell}(t)}{2}\right) \cdot \mu(t) \right)                 \label{equation:sineAcos} \\
	                                     \\
		& \leq \sigma(t)  \, \left( \frac{\sqrt{3}}{2} \cdot \sin \left(\beta^{+}(t)\right) - \cos \left(\frac{\alpha_{\ell}(t)}{2}\right) \cdot \mu(t) \right)                                                    \\
                          \label{equation:rewriteAlphaEll}            \\
		& \leq  \sigma(t)  \, \left( \frac{\sqrt{3}}{2} \cdot \sin \left(c\right) - \sin \left(c\right) \cdot \cos \left(\frac{c}{2} - \beta^{+}(t)\right)  \right)             \label{equation:cosineVodoo}       \\
		& \leq  \sigma(t)  \, \left( \frac{\sqrt{3}}{2} \cdot \sin \left(c\right) - \sin \left(c\right) \cdot \cos \left(\frac{c}{2}\right)  \right)                                                               = \sigma(t) \, \sin \left(c\right) \ \cdot  \left(  \frac{\sqrt{3}}{2}  - \cos \left(\frac{c}{2} \right)  \right)       < 0
                                                                                     \\
	\end{align}

\end{proof}

However, we cannot use $I(t)$ as a progress measure here, because for very large angles in the order of $\pi-\frac{1}{n}$, $I(t)$ does not decrease with constant speed anymore.
Therefore, we introduce another progress measure which decreases with constant speed in this case (see \Cref{figure:height}).

\begin{definition}
	Define $H_\ell(t)$ to be the distance of $r_{\ell(t)}$ to the line segment connecting $r_1$ and $r_{\ell^{+}(t)}$
	and define $H_r(t)$ to be the distance of $r_{r(t)}$ to the line segment connecting $r_{r^{+}(t)}$ and $r_{n}$.
\end{definition}

\begin{figure}[hbt]
	\begin{minipage}[hbt]{0.45\textwidth}
		\centering
		\vspace*{0.4cm}
		\includegraphics[width=1\textwidth, clip=true]{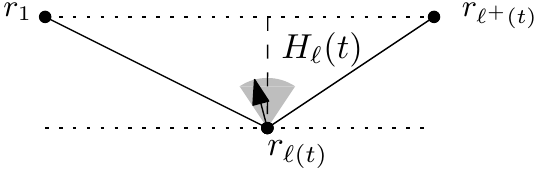}
		\vspace*{0.7cm}
		\caption{Visualization of  $H_{\ell}(t)$. The gray part marks the area of all possible vectors $v_{\ell}(t)$.}
		\label{figure:height}
	\end{minipage}
	\hfill
	\begin{minipage}[hbt]{0.48\textwidth}
		\begin{minipage}[hbt]{0.45\textwidth}
			\centering
			\includegraphics[width=\textwidth]{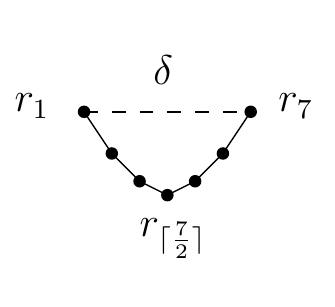}
			\subcaption{
				A discrete \hfill\\ \( \delta \)-V-configuration.
			}\label{fig:discreteThetaV}
		\end{minipage}
		\phantom{2}
		\begin{minipage}[hbtt]{0.45\textwidth}
			\centering
			\includegraphics[width=\textwidth]{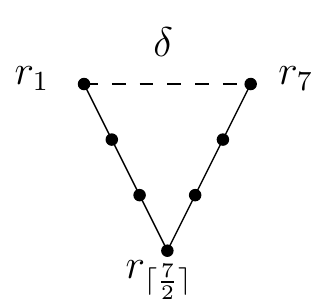}
			\subcaption{A continuous \hfill\\ \( \delta \)-V-configuration.}\label{fig:contThetaV}
		\end{minipage}
		\caption{%
			A depiction of discrete \( \delta \)-V-con-\\figurations and continuous \( \delta \)-V-configu-\\rations.}\label{fig:bothThetaVs}
	\end{minipage}
\end{figure}

\begin{restatable}{lemma}{lemmaMaxMobHeight} \label{lemma:outerAngleHeightNew}
	In configurations fulfilling $\ell(t) \neq r(t)$, $\alpha_{i}(t) \geq \frac{3}{4}\pi$ for $i \in \{\ell(t), r(t)\}$ and $O_i(t) = \gamma_i(t)$ it holds that  $\deriv{H_i}{t} \leq - \frac{1}{20}$.
\end{restatable}

\begin{proof}
	Assume that $i = \ell(t)$, the proof for $i = r(t)$ is analogous.
	We have to analyze the movements of $r_1, r_{\ell(t)}$ and $r_{\ell^{+}(t)}$ in this case.
	Rewrite $\alpha_{\ell}(t) = \pi - c$ for $c \leq \frac{\pi}{4}$.
	Without loss of generality assume $r_{\ell(t)}$ to be positioned in the origin and the line segment connecting $r_1$ and $r_{\ell^{+}(t)}$ to be a parallel line to the $x$-axis above of $r_{\ell(t)}$.
	As all robots between $r_1$ and $r_{\ell(t)}$ as well as all robots $r_{\ell(t)}, \dots r_{\ell^{+}(t)}$ are collinear, $v_{\ell}(t)$ must point upwards.
	Since $\alpha_{\ell}(t) \geq \frac{3}{4} \pi$, $v_{\ell}(t)$ must form an angle of size less than $\frac{\pi}{8}$ with the $y$-axis.
	Hence, $r_{\ell(t)}$ moves with speed at least $\cos \left(\frac{\pi}{6}\right) > 0.92$ upwards.
	At the same time the robots $r_1$ and $r_{\ell^{+}(t)}$ could move.
	Consider the movement of $r_1$.
	Similar to the proof of \Cref{lemma:outerAngleIncreasing} (see also \Cref{figure:triangleHeight}), let $\beta^{-}(t)$ be the angle formed by vectors pointing from $r_1$ to $r_{\ell(t)}$ and from $r_1$ to $r_{\ell^{+}(t)}$ and let $\beta^{+}(t)$ be the angle formed by vectors pointing from $r_{\ell^{+}(t)}$ to $r_1$ and from $r_{\ell^{+}(t)}$ to $r_\ell$.
	$ v_1(t) = - \cos \left(\frac{\alpha_{i}(t)}{2}\right) \cdot \widehat{w}_2(t)$.
	Thus, $\beta_{1, \ell}(t) = \pi - \beta^{-}(t)$.
	Therefore, $r_1$ moves upwards with speed $\sin \left(\beta^{-}(t)\right) \cdot \cos \left(\frac{\alpha_{\ell}(t)}{2}\right) = \sin \left(\beta^{-}(t)\right) \cdot \cos \left( \frac{\pi}{2} - \frac{\alpha_{\ell}(t)}{2}\right) = \sin \left(\beta^{-}(t)\right) \cdot \sin \left(\frac{c}{2}\right)$.
	Lastly observe $\beta^{-}(t) \leq c$ as the sum of internal angles of a triangle is $\pi$.
	Hence, $r_1$ moves upwards with speed at most $\sin \left(c\right) \cdot \sin \left(\frac{c}{2}\right) \leq \sin \left(\frac{1}{4}\pi\right) \cdot \sin \left(\frac{\pi}{8}\right) < 0.28$.
	It remains to analyze the speed of $r_{\ell^{+}(t)}$ moving upwards.
	As $r(t) \neq \ell(t)$, $r_{\ell^{+}(t)}$ either does not move at all or $\alpha_{\ell^{+}}(t) < \referenceAngle{}$.
	In case $r_{\ell^{+}(t)}$ does not move at all $H_{\ell}(t) <= -0.92 + 0.28 = 0.64$.
	Now consider $\alpha_{\ell^{+}}(t) < \referenceAngle{}$.
	$\beta^{+}(t)$ could be almost $0$, such that in the worst case the angle formed by $v_{\ell^{+}}(t)$ and the line segment connecting $r_1$ and $r_{\ell^{+}(t)}$ lies completely above that line segment.
	In this case, $r_{\ell^{+}(t)}$ moves upwards with speed at most $\sin \left(\frac{\referenceAngle{}}{2}\right) = \sqrt{1-\left(1-\tau\right)^2} = \sqrt{1-1+2\cdot \tau - \tau^2} = \sqrt{2 \cdot \tau - \tau^2} \leq \sqrt{2 \cdot \frac{1}{2} - \frac{1}{4}} = \frac{\sqrt{3}}{2} < 0.87$.
	Therefore,  $H_{\ell}(t) \leq -0.92 + 0.87 = - 0.05 = - \frac{1}{20}$.
\end{proof}

By noticing that $H_{i}(t)$, for $i \in \{\ell(t), r(t)\}$, is upper bounded by $\Delta_{i,i^{+}}(t) \leq |i-i^{+}(t)| $ and using the fact that $H_i(t)$ cannot decrease anymore since $\alpha_{i}(t)$ is monotonically increasing, we can derive the total time an outer angle can have a size of at least $\frac{3}{4} \pi$ while $O_i(t) = \gamma_i(t)$.

\begin{corollary} \label{corollary:largeSize}
	The total time an outer angle of size at least $\frac{3}{4}\pi$ exists, while $O_i(t) = \gamma_i(t)$ and $\ell(t) \neq r(t)$, is bounded by $20 \cdot \left(n-3\right)$.
\end{corollary}

A combination of the preceding insights leads to a time bound until $\ell(t) = r(t)$.

\begin{restatable}{lemma}{lemmaMaxMobMainLemmaFirst} \label{lemma:mainLemmaFirst}
	After time at most  $2 \cdot \left(n-3\right) \cdot \left(\frac{1}{1-\tau} + \frac{1}{\sqrt{2-\sqrt{2}}} + 10\right)$ it holds that $\ell(t) = r(t)$.
\end{restatable}

\begin{proof}
	The total time in configurations with an outer angle of size at most \referenceAngle{} is bounded by $\frac{n-3}{1-\tau}$ (\Cref{lemma:turnaroundTotal}).
	It remains to bound the time of configurations with larger outer angles.
	As soon as an outer angle reaches a size of at least \referenceAngle{}, $O_i(t)$ increases with speed $1-\tau$ (\Cref{lemma:outerAngleMediumVectors}) until it reaches maximal length.
	Additionally, by \Cref{lemma:outerAnglesSmallVectors}, $O_i(t)$ will only decrease in the future in case the corresponding outer angle has a size of less than \referenceAngle{} (\Cref{lemma:outerAnglesSmallVectors}).
	Hence, the total time all $O_i(t)$ can decrease is bounded by $\frac{n-3}{1-\tau}$, the total time outer angles can have a size of less than \referenceAngle{}.
	In case an outer angle $\alpha_{i}(t)$ has a size of at least \referenceAngle{} and $O_i(t) < \gamma_i(t)$ it holds that $O_i(t)$ increases with speed $1- \tau$ (\Cref{lemma:outerAngleMediumVectors}).
	As $O_i(t)$ is bounded by $n-3$ while $\ell(t) \neq r(t)$, the total time all $O_i(t)$ can increase with speed $1-\tau$ is bounded by $\frac{n-3}{1-\tau}$.
	Furthermore, the total time $O_i(t) = \gamma_i(t)$ and $\alpha_{i}(t) \leq \frac{3}{4} \pi$ is bounded by $\frac{2 \cdot \left(n-3\right)}{\sqrt{2-\sqrt{2}}}$ (\Cref{corollary:mediumSizeTime}).
	The total time $O_i(t) = \gamma_i(t)$ and $\alpha_{i}(t) \geq \frac{3}{4} \pi$ is bounded by $20 \cdot \left(n-3\right)$.
	Hence, the total time needed until $\ell(t) = r(t)$ is bounded by $\frac{2 \cdot \left(n-3\right)}{1-\tau} + \frac{2 \cdot \left(n-3\right)}{\sqrt{2-\sqrt{2}}}$ + $20 \cdot \left(n-3\right) = 2 \cdot \left(n-3\right) \cdot \left(\frac{1}{1-\tau} + \frac{1}{\sqrt{2-\sqrt{2}}} + 10\right)$.
\end{proof}

Lastly, it remains to analyze the case $\ell(t) = r(t)$.
The combination of \Cref{lemma:mainLemmaFirst,lemma:mainLemmaSecond} yields then a total runtime bound.

\begin{restatable}{lemma}{lemmaMaxMobMainLemma} \label{lemma:mainLemmaSecond}
	Assume that $\ell(t)  = r(t)$.
	Then, after time $3n \cdot \left(\frac{1}{\tau} +\frac{1}{1-\tau}\right)$, the configuration is transformed into a \maxChain{} or all robots are located on the same position.
\end{restatable}

For the proof of \Cref{lemma:mainLemmaSecond}, we need some lemmata dealing with one-dimensional configurations.

\begin{lemma} \label{lemma:oneDimLeftEqualsRight}
	Starting \textnormal{\parameterizedMaxMoveOnBisector{1}{1-\tau}} in a one-di\-men\-sional configuration, $\ell(t) \neq r(t)$ can hold for time at most $\frac{n-3}{2}$.
\end{lemma}

\begin{proof}
	This is a conclusion of \Cref{lemma:smallAnglesInnerLength}.
	In this configuration, both $\alpha_{\ell}(t) = 0$ and $\alpha_{r}(t) = 0$ and $\norm{v_\ell(t)} = \norm{v_r(t)} = 1$.
	Thus, $\innerLength{}$ decreases with speed at least $\cos \left(\frac{\alpha_{\ell}(t)}{2}\right) +\cos \left(\frac{\alpha_{r}(t)}{2}\right) = 2 \cdot \cos \left(0\right)~= ~2$.
\end{proof}

\begin{restatable}{lemma}{lemmaMaxMobOpposedOneDim} \label{lemma:oneDimOpposed}
	A one-dimensional configuration fulfilling $\ell(t) = r(t) = 0$ is transformed into a \maxChain{} after time at most $\frac{n-1}{2 \cdot \left(1-\tau\right)}$.
\end{restatable}

\begin{proof}
	In case $\ell(t) = r(t) = 0$, it holds $\alpha_{\ell}(t) = \pi$ and thus all vectors point into the same direction rendering the configuration an \opposedConfig{}.
	In \opposedConfig{}s, the outer robots move both with speed $\left(1-\tau\right)$ away from each other such that their distance decreases with speed $2 \cdot \left(1-\tau\right)$.
	Since the maximal distance is $n-1$, a straight line of length $n-1$ is obtained after time at most $\frac{n-1}{2 \cdot \left(1-\tau\right)}$.
\end{proof}

\begin{restatable}{lemma}{lemmaMaxMobOneDimOpposing} \label{lemma:oneDimMarching}
	A one-dimensional configuration fulfilling $\ell(t) = r(t) = x$ for $1 < x < n$ and $p_1(t) \neq p_n(t)$ is transformed into a \maxChain{} after time at most $\frac{n-1}{2} \cdot \left(\frac{1}{\tau} + \frac{1}{1-\tau}\right)$.
\end{restatable}

\begin{proof}
	In such a configuration it holds $\leftOuter{} \neq \rightOuter{}$ and $\alpha_{\ell}(t) = 0$.
	The robot $r_{\ell(t)}$ moves with speed $1$ towards the two outer robots.
	The outer robots move with speed $1-\tau$ away from $r_{\ell(t)}$.
	These movements cause a decrease of both $\leftOuter{}$ and $\rightOuter{}$ with speed $\tau$.
	To see this, note that $\beta_{1, \ell}(t) = \beta_{n,\ell}(t) = \pi$ and $\beta_{\ell, 1}(t) = \beta_{\ell,n}(t) = 0$.
	\Cref{lemma:continuousDistances} yields $\deriv{\Delta_{1,\ell}}{t} = \deriv{O_\ell}{t} = -  \norm{v_1(t)} $ $ \cdot \cos \beta_{1, \ell}(t) - 1 \cdot \cos \beta_{\ell, 1}(t) = - \left(1-\tau\right) +1 = - \tau$.
	As a consequence, after time at most $\frac{n-1}{2\cdot\tau}$ either $\leftOuter{}$ or $\rightOuter{}$ reaches size $0$ such that in the following $\ell(t) = r(t) = 0$.
	Such a configuration is transformed into a straight line of length $n-1$ after time at most $\frac{n-1}{2 \cdot \left(1-\tau\right)}$ (\Cref{lemma:oneDimLeftEqualsRight}).
\end{proof}

\begin{restatable}{lemma}{lemmaMaxMobOneDimMarching} \label{lemma:oneDimGather}
	A one-dimensional configuration fulfilling $\ell(t) = r(t) = x$ for $1 < x < n$ and $p_1(t) = p_n(t)$ is transformed into a configuration in which all robots are located on the same position after time at most $\frac{n-1}{2 \cdot \tau}$.
\end{restatable}

\begin{proof}
	Since the outer robots are located at the same position it must hold $\leftOuter{} = \rightOuter{}$ and $\alpha_{\ell}(t) = 0$.
	$r_{\ell(t)}$ moves with speed $1$ towards the two outer robots.
	In this configuration, it holds $\beta_{1, \ell}(t) = \beta_{n, \ell}(t) = \pi$ and $\beta_{\ell,1}(t) = \beta_{\ell, n}(t) = 0$.
	Combined with $\norm{v_1(t)} = \norm{v_n(t)} = 1 - \tau$ and $\norm{v_{\ell}(t)} = 1$, \Cref{lemma:continuousDistances} gives us
	$\deriv{\Delta_{1,\ell}}{t} = \deriv{\leftOuterShort{}}{t} = - \left(- \left(1-\tau\right) +1\right) = - \tau$.
	As $\leftOuter{}$ and $\rightOuter{}$ are both bounded by $\frac{n-1}{2}$, the lemma follows.
\end{proof}

\begin{proof}[Proof of \Cref{lemma:mainLemmaSecond}]
	There are two cases in which $\ell(t) = r(t)$.
	Either $\ell(t) = r(t) = 0$ or $\ell(t) = r(t) = j$ with $2 \leq j \leq n-1$.
	In case $\ell(t) = r(t) = 0$ this means that either every vector is the $0$-vector and thus all robots are located on a single point or all robots are located on the same line in an \opposedConfig{}.
	This is a one-dimensional configuration and transformed into a line of length $n-1$  after time at most $\frac{n}{2 \cdot \left(1-\tau\right)}$ (\Cref{lemma:oneDimOpposed}).
	It remains to consider $\ell(t) = r(t) = j$ with $2 \leq j \leq n-1$.
	In case $\alpha_{\ell}(t) = 0$ the configuration is one-dimensional and transformed into a straight line or a single point after time at most $\frac{n-1}{2} \cdot \left(\frac{1}{\tau} + \frac{1}{1-\tau}\right)$ (\Cref{lemma:oneDimMarching,lemma:oneDimGather}).
	Assume that $\alpha_{\ell}(t) > 0$.
	There is only a single angle of size less than $\pi$ in this configuration.
	For the special case of a \thetaV{} we have proven a runtime of at most $n \cdot \left(\frac{1}{\tau} + \frac{1}{1-\tau}\right)$ in \Cref{def:continuousThetaV}.
	Now, suppose the triangle formed by $r_1, r_{\ell(t)}$ and $r_n$ is not isosceles.
	W.l.o.g.\ assume that $\Delta_{\ell,n}(t) < \Delta_{1,\ell}(t)$.
	In this case, enlarge the line segment connecting $r_{\ell(t)}$ and $r_n$ such that it has length $\Delta_{1,\ell}(t)$ and place a virtual robot $r_v$ at the end of this line segment.
	Now, the triangle formed by $r_1$, $r_\ell(t)$ and $r_v$ is an isosceles triangle.
	Assume that the virtual robot $r_v$ moves exactly as $r_{n}$.
	Define $H_v(t)$ to be the distance of $r_{\ell(t)}$ to the line segment connecting $r_1$ and $r_{v}$.
	$r_{\ell(t)}$ moves with speed $1$ upwards while both $r_1$ and $r_v$ can move with speed at most $1-\tau$ upwards.
	The rest of the argumentation is analogous to \thetaVs{} with the only difference that $H_{\ell}(t)$ is bounded by $n-2$ (in case $\leftOuter{} = n-2$ and $\rightOuter{} = 1$ or vice versa).
	Thus the total time spent in such a configuration can be bounded by $2n \cdot \left(\frac{1}{\tau} + \frac{1}{1-\tau}\right) + \frac{n}{2 \cdot \left(1-\tau\right)} < 3n \cdot \left(\frac{1}{\tau}  + \frac{1}{1-\tau} \right)$.
\end{proof}

\theoremMainTheorem*

\Cref{lemma:mainLemmaSecond} states that there might be two-dimensional configurations in which the chain contracts to a single point instead of reaching the \maxChain{}.
Our simulations support the following conjecture.

\conjectureLebesgueContinuousTwo*

\section{The Influence of the Speed of outer Robots}\label{section:Influence-of-outer-robots-extended}

We conclude by discussing the role of the speeds of outer robots in \parameterizedMaxMoveOnBisector{1}{1-\tau} and \AName{}, revealing an interesting runtime gap between the continuous time model and \fsync{}.
There are two classes of configurations that play an important role in this discussion, \emph{discrete} $\delta$-V-configurations (\Cref{definition:2-dimensions:delta-u-configurations}) and \emph{continuous} $\delta$-V-configurations (\Cref{def:continuousThetaV}, see below).
Both configurations are depicted in \Cref{fig:bothThetaVs}.
With help of these configurations, we give evidence why the speed of outer robots is reduced to $1-\tau$ in \parameterizedMaxMoveOnBisector{1}{1-\tau}. Continuous $\delta$-V-configurations are resolved by \maxmob/ in time $\mathcal{O}\left(n\right)$.

\begin{definition} \label{def:continuousThetaV}
	Let $\delta$ be a positive constant and define $\theta = 2 \cdot \sin^{-1}\left(\delta/ \lfloor \frac{n}{2} \rfloor\right)$.
	For $n$ odd, the  \emph{\thetaV{}} forms an isosceles triangle with $\norm{w_i(t)} = 1$ for all $2 \leq i \leq n$, $\alpha_{\lceil \frac{n}{2} \rceil}(t) = \theta$ and for all other angles it holds $\alpha_{i}(t) = \pi$.
\end{definition}

\lemmaMaxMobContinuousThetaV*

\begin{proof}
	Fix a point in time $t_0$ in which the configuration forms a \thetaV{}.
	Note that a \thetaV{} forms an isosceles triangle whose legs have a length of $\lfloor \frac{n}{2} \rfloor$ and the base ($\Delta_{1,n}(t_0)$) has a length of $\delta$.
	W.l.o.g.\ assume $p_{\ell}(t_0) = (0,0)$, $p_1(t_0) = (x_1 ,y_1)$ and and $p_{n}(t_0) = (x_{n},y_{n})$  with $x_1 = - \frac{\delta}{2}, y_1 = \lfloor \frac{n}{2} \rfloor \cdot \cos \left(\frac{\theta}{2}\right),$ $x_n = \frac{\delta}{2}$ and $y_n = \lfloor \frac{n}{2} \rfloor \cdot \cos \left(\frac{\theta}{2}\right)$.
	Consider the case $\alpha_{\ell}(t) = \theta < \referenceAngle{}$.
	$r_{\ell(t)}$ moves with speed $1$ upwards.
	As $r_1$ and $r_{n}$ move with speed at most $1-\tau$, $H_{\ell}(t)$ and $H_r(t)$ decrease with speed at least $\tau$.
	As $H_{\ell}(t_0) = \lfloor \frac{n}{2} \rfloor \cdot \cos \left(\frac{\theta}{2}\right) < \frac{n}{2}$, $\alpha_{\ell}(t) \geq \referenceAngle{}$ must hold after time at most $\frac{n}{2 \cdot \tau}$, otherwise $H_{\ell}(t) = 0$ and, thus, $\alpha_{\ell}(t) = \pi$.
	Since $\alpha_{\ell}(t) < \referenceAngle{}$ initially and $\alpha_{\ell}(t)$ changes continuously,  $\alpha_{\ell}(t) \geq \referenceAngle{}$  must hold before $H_{\ell}(t) = 0$.
	As soon as $\alpha_{\ell}(t)$ reaches a size of $\referenceAngle{}$, $r_{\ell(t)}$ stops moving as its movement has decreased $\leftOuter{}$ and $\rightOuter{}$.
	Then, both $\leftOuter{}$ and $\rightOuter{}$ increase with speed $1-\tau$ (\Cref{lemma:outerAngleMediumVectors}).
	As $\leftOuter{}$ and $\rightOuter{}$ are bounded by $\lfloor\frac{n}{2}\rfloor$ it holds $\leftOuter{} = \gamma_\ell(t)$ and $\rightOuter{} = \gamma_r(t)$ after time at most $\frac{1}{1-\tau} \cdot \lfloor \frac{n}{2} \rfloor$.
	Afterwards, $r_{\ell(t)}$ continues moving with speed $1$ upwards, decreasing $H_{\ell}(t)$ with speed at least $\tau$ (we can apply the same arguments as before).
	Thus, finally after additional time of at most $\frac{n}{2 \cdot \tau}$ it holds $H_\ell(t) = 0$ and the configuration is a one-dimensional \opposedConfig{} that is transformed into a straight line of length $n-1$ after time at most $\frac{n-1}{2 \cdot \left(1-\tau\right)}$ (\Cref{lemma:oneDimOpposed}).
	We conclude that the total time is upper bounded by $\frac{2 \cdot n}{2 \cdot \tau} + \frac{n}{2} \cdot \frac{1}{1-\tau} + \frac{n-1}{2 \cdot \left(1-\tau\right)} < n \cdot \left(\frac{1}{\tau} + \frac{1}{1-\tau}\right)$.
\end{proof}

Consider now the na\"ive approach that the outer robots always move at full speed in \maxmob/.
We call this strategy \naivemaxmob/.
We prove that the runtime of \naivemaxmob/ for \thetaVs{} depends also on $\delta$, exactly as in the lower bound for the discrete case.
Thus, for flattening the chain, it is crucial that the inner robots move significantly faster upwards than the outer robots in these configurations.

\theoremOneOneMaxMob*

For the proof of \Cref{theorem:thetaTrianglesRuntime}, we state the following lemma:

\begin{restatable}{lemma}{lemmaMaxMobFullSpeedThetaV} \label{lemma:continuousThetaVSpeed}
	When applying \textnormal{\naivemaxmob/} to \thetaVs{}, $$\deriv{ \Delta_{1,n}}{t}=\Delta_{1,n}(t) \cos \bigl(\frac{\theta}{2}\bigr)/\lfloor \frac{n}{2} \rfloor. $$
\end{restatable}

\begin{proof}
	Due to the symmetry, we obtain
	$\beta_{1,\ell}(t) = \beta_{n,\ell}(t) = \frac{\pi}{2} + \frac{\theta}{2}$.
	In \thetaVs{}, the distances $ \Delta_{\ell,\ell^{+}}(t)$ and $ \Delta_{r,r^{+}}(t)$ remain constant, because the outer robots are able to move with speed $1$.
	Hence, the outer robots move with speed $\cos \left(\frac{\theta(t)}{2}\right)$ as the robot $r_{\lceil \frac{n}{2} \rceil}$ reduces $ \Delta_{\ell,\ell^{+}}(t)$ and $ \Delta_{r,r^{+}}(t)$ with speed $\cos \left(\frac{\theta(t)}{2}\right)$.
	Thus, we can calculate $\deriv{ \Delta_{1,n}}{t}$ according to \Cref{lemma:continuousDistances} as follows:
	$\deriv{ \Delta_{1,n}}{t} = - 2\cdot \left( \cos \left(\frac{\theta}{2}\right) \cdot \cos \left(\frac{\pi}{2} + \frac{\theta}{2}\right) \right) = - 2\cdot \left(- \cos \left(\frac{\theta}{2}\right) \cdot \sin \left(\frac{\theta}{2}\right) \right)                 = \sin \left(\theta\right)$.
	Via the law of sines, we then obtain: $\frac{ \Delta_{1,n}(t)}{\sin\left(\theta\right)} = \frac{\lfloor\frac{n}{2} \rfloor}{\sin \left(\frac{\pi - \theta}{2}\right)}
	\iff \sin\left(\theta\right) = \frac{ \Delta_{1,n}(t) \cdot \cos \left(\frac{\theta}{2}\right)}{\lfloor \frac{n}{2} \rfloor} $

\end{proof}

\begin{proof}[Proof of \cref{theorem:thetaTrianglesRuntime}]
	Fix a point $t_0$ such that $\Delta_{1,n}(t_0) = \delta$, according to the definition of \thetaVs{}.
	Since $\cos (\theta(t)/2) \leq 1$, we can bound $\deriv{ \Delta_{1,n}}{t} \leq \frac{2 \cdot  \Delta_{1,n}(t)}{n}$ (see \Cref{lemma:continuousThetaVSpeed}).
	Thus, it requires time $\mathcal{O}(n)$ until $ \Delta_{1,n}(t)$ doubles.
	To increase $ \Delta_{1,n}(t)$ such that $ \Delta_{1,n}(t) \geq c$ for an arbitrary constant (less than $1$), it requires time $\Omega \left(n \cdot \log \left(1/ \Delta_{1,n}(t_0)\right)\right) = \Omega \left(n \cdot \log \left(1/ \delta\right)\right)$.
\end{proof}

This also explains another aspect of \maxmob/.
An inner robot moves in case either $\norm{w_{i}(t) } = 1$, $\norm{w_{i+1}(t)} = 1$ or $\alpha_{i}(t) < \referenceAngle{}$.
Suppose we drop the last assumption and inner robots move only in case either $\norm{w_{i}(t) } = 1$, $\norm{w_{i+1}(t)} = 1$.
This has the consequence that we lose the speed gain obtained by reducing the speed of the outer robots!
To see this, observe that in a \thetaV{} with very small angles $\theta$, the robot $r_{\lceil \frac{n}{2} \rceil}$ moves fast enough such that $O_{\ell}(t)$ and $O_{r}(t)$ decrease with constant speed such that immediately $\norm{w_{\lceil \frac{n}{2} \rceil}(t) } < 1$ and $\norm{w_{\lceil \frac{n}{2} \rceil+1}(t)} < 1$ hold.
Hence, $r_{\lceil \frac{n}{2} \rceil}$ stops moving and waits until $O_{\ell}(t)$ and $O_{r}(t)$ reach their maximum length again.
As the process is continuous, $r_{\lceil \frac{n}{2} \rceil}$ does not wait until this happens but is slowed down to a speed of $1-\tau$ such that $r_{\lceil \frac{n}{2} \rceil}$ and the outer robots move with the same speed that results in a runtime depending on $\delta$.
To summarize, two aspects in the design of \maxmob/ are crucial for the linear runtime: Slowing down the outer robots to a speed of $1-\tau$ and allowing the inner robots to move with full speed along their bisectors in case they are located at a very small angle.

Since slowing down the outer robots in the continuous time model removes the dependence one $\delta$, one could conjecture that the same approach would also work in the discrete time model.
Consider the strategy $\left(1-\tau\right)$-\AName{}, in which the outer robots do not move the full distance to their target point but only $1-\tau$ times the distance they would usually move.
The movement of inner robots remains unchanged.
The new positions of $r_1$ and $r_n$ can be computed as follows:
$p_1(t+1)  = p_1(t) + \oneminustau \, \left(\half \cdot p_{2}(t) + \half p_1(t) - \half \widehat{w}_2(t) - p_1(t) \right)
= \halfoneplustau \cdot p_1(t) + \halfoneminustau \cdot p_{2}(t) - \halfoneminustau \cdot \widehat{w}_2(t)$ and
$p_n(t+1)  = \halfoneplustau \cdot p_n(t) + \halfoneminustau \cdot p_{n-1}(t) + \halfoneminustau \cdot \widehat{w}_n(t)$.
For the vector representation, we obtain the following equations:
$w_2(t+1) =  \half w_3(t) + \tauhalf w_2(t) + \halfoneminustau  \cdot \widehat{w}_2(t)$ and
$w_n(t+1)  = \half \cdot w_{n-1}(t) + \tauhalf \cdot w_n(t) + \halfoneminustau  \cdot \widehat{w}_n(t)$

Similar to \Cref{definition:2-dimensions:delta-u-configurations}, we can define configurations that have the same behavior under \tauAName{} as \deltaUConfig{}s under \AName{} showing that a speed reduction does not work here.

\begin{restatable}{definition}{definitionOneMinusTauThetaV} \label{definition:2-dimensions:tau-delta-u-configurations}
	For $n$ even, a \emph{\taudeltaUConfig{}} is defined by the vectors $w_2(t) = \left(\frac{\delta}{n-1},\frac{1-\tau}{1-\tau + \frac{2}{n-2}} \right)$ and for all $2 < i \leq n$:
	$w_i(t) = \left(\frac{\frac{n}{2} - i+1}{\frac{n}{2}-1}\right) \cdot w_2(t)$.
\end{restatable}

Note that for $\tau = 0$, \deltaUConfig{}s and \taudeltaUConfig{}s coincide.
Also for $\delta = 0$, \taudeltaUConfig{} have a marching chain behavior, however the movement distance per round scales with $\tau$.

\theoremOneMinusTauGtmLower*

\begin{lemma} \label{lemma:2-dimension:tau-lower-bound-monotonically-increasing-outer-vectors}
	During an execution of \textnormal{\tauAName{}} starting in a \taudeltaUConfig{} at time step $t_0$, it holds that $\norm{w_i(t)} \geq \norm{w_i(t_0)} $ for all $t$ and all $2 \leq i \leq n$.
\end{lemma}

\begin{proof}
	The proof is analogous to the proof of \Cref{lemma:2-dimension:lower-bound-monotonically-increasing-outer-vectors}.
\end{proof}

\begin{restatable}{lemma}{lemmaMaxGtmTauLower} \label{lemma:maxGtmTauLower}
	Assume that we start \tauAName{} in a \taudeltaUConfig{}.
	It holds $x_2(t) \geq \half$ after $\Omega\left(n^2 \cdot \log \left(1/\delta\right)\right)$ rounds.
\end{restatable}

\begin{proof}
	\begin{align*}
		x_2(t+1)  \leq  \half x_{3}(t) + \tauhalf \cdot x_2(t) + \frac{1 - \tau}{2 \cdot \frac{1-\tau}{1 - \tau + \frac{2}{n-2}}} \cdot x_2(t)
		& = \half x_3(t) + \left(\half + \frac{1}{n-2}\right) \cdot x_2(t) \\
        & \leq \left(1+ \frac{1}{n-2}\right) \cdot x_2(t)
	\end{align*}

	The last line is exactly the same formula that has been obtained in the proof of \Cref{lemma:2-dimensional-lower-bound-x} and thus, the same lower bound holds.
\end{proof}




%
%

\end{document}